\documentclass[a4paper,12pt]{article}
\usepackage{amsthm}
\usepackage[utf8x]{inputenc}
\usepackage{graphicx}
\usepackage{xcolor}

\usepackage{setspace}
\usepackage{fullpage}
\usepackage{multirow}
\usepackage{amssymb}
\usepackage{subfigure}
\usepackage{amsmath}
\usepackage{subfigure}

\usepackage{lineno}

\newtheorem{theorem}{Theorem}
\newtheorem*{theorem*}{Theorem}
\newtheorem{lemma}{Lemma}

\newtheorem{proposition}[theorem]{Proposition}
\newtheorem{definition}{Definition}
\newtheorem{conjecture}{Conjecture}
\newcommand{\note}[1]{}

\DeclareMathOperator{\nullity}{nullity}

\renewcommand{\epsilon}{\varepsilon}
\renewcommand{\subset}{\subseteq}
\newcommand{\red}{\textsf{red}}
\newcommand{\blue}{\textsf{blue}}

\newcommand{\allred}{S_R}
\newcommand{\allblue}{S_B}

\newcommand{\enemies}[2]{N'_{#1}(#2)}

\newcommand{\byflipping}[2]{{#1} \oplus {#2}}

\newcommand{\statespace}{\mathbb{S}}
\newcommand{\transmatrix}{\mathbf{P}}
\newcommand{\bgcolor}{b}

\renewcommand{\vec}[1]{\mathbf{#1}}

\title{Environmental Evolutionary Graph Theory}
\author{Wes Maciejewski, Gregory J.~Puleo}

\begin{document}
 
 \maketitle
 \linenumbers
 \begin{abstract}
   Understanding the influence of an environment on the evolution of
   its resident population is a major challenge in evolutionary
   biology. Great progress has been made in homogeneous population
   structures while heterogeneous structures have received relatively
   less attention. Here we present a structured population model where
   different individuals are best suited to different regions of their
   environment. The underlying structure is a graph: individuals
   occupy vertices, which are connected by edges. If an individual is
   suited for their vertex, they receive an increase in fecundity. This
   framework allows attention to be restricted to the spatial
   arrangement of suitable habitat. We prove some basic properties of
   this model and find some counter-intuitive results. Notably, 1) the
   arrangement of suitable sites is as important as their proportion,
   and, 2) decreasing the proportion of suitable sites may result in a
   decrease in the fixation time of an allele.
 \end{abstract}

 \doublespacing
 
\section{Introduction}

It is now well established that population structure can have a profound effect on the outcome of an evolutionary process. Indeed, some of the first results in the modern synthesis of evolution considered island-structured populations \cite{wright31}. Since then, a multitude of structured population models have appeared, including stepping stone \cite{kimuraweiss64}, lattice \cite{nowakmay92, nakamarumatsudaiwasa97}, and metapopulation models \cite{levins69}. A contemporary take on these spatial models is evolutionary graph theory. 

Since its introduction in \cite{liebermanhauertnowak05}, evolutionary
graph theory has gone on to become a well-studied abstraction of
structured populations (see \cite{nowak06a} for an illustrative
introduction and \cite{szabofath07} for an extensive review). An
evolutionary graph is a collection of sites, or \emph{vertices},
linked by interaction and dispersal patterns, or \emph{edges}. Each
vertex is occupied by a single haploid breeder of a certain genotype
-- say, red or blue. Lieberman, Hauert, and
Nowak~\cite{liebermanhauertnowak05} considered a population of
blue-type individuals invaded by a single red individual of higher
fecundity. Subsequent work considered strategic interactions between
the residents of a graph, where ``red'' and ``blue'' are thought of as
the strategies adopted by the individuals. This perspective proved
useful, and evolutionary graphs have gone on to facilitate much
understanding in evolutionary game theory in structured populations
\cite{ohtsukihauertliebermannowak06, taylordaywild07a}. 

Here we introduce environmental evolutionary graph theory as a variant
on evolutionary graph theory. An environmental evolutionary graph is a
graph with vertices of different types. We now assign colours not only
to \emph{individuals}, but also to the vertices of the graph. We
typically consider a two-colour setup: each individual is either red or
blue, and each vertex of the graph is also either red or blue,
independent of the colour of the individual occupying that vertex.  An
individual whose colour matches the vertex on which it resides is given
a higher fecundity, reflecting the individual's adaptation to that
particular environment. We formalize the model in the appendix.

\begin{figure}
\centering
  \subfigure[]{\label{fig:latticeex}\includegraphics[ width=0.3\textwidth]{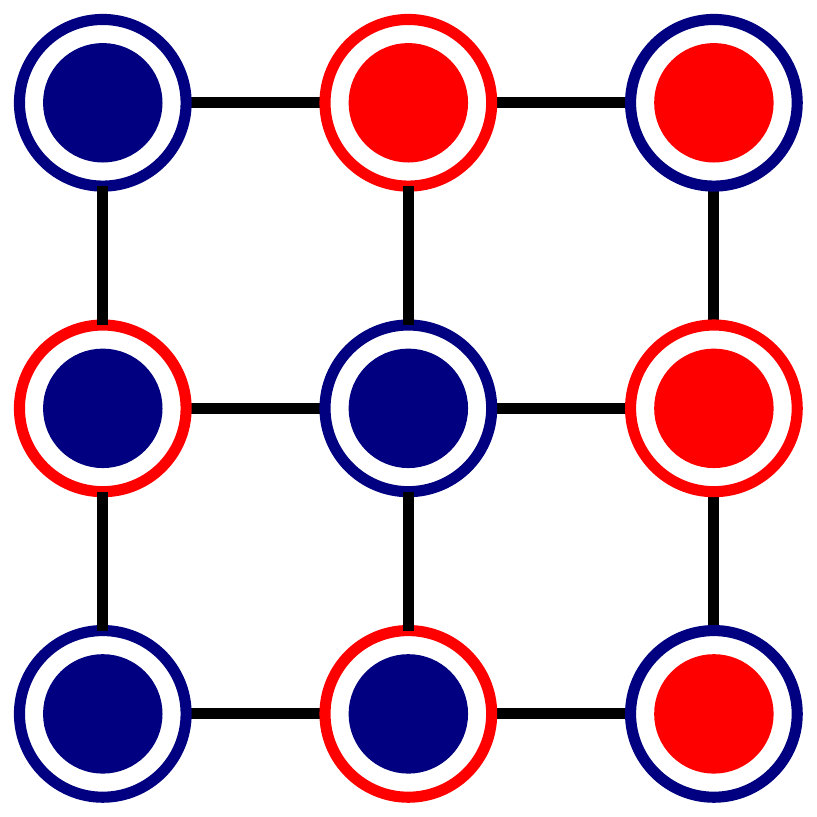}}
 \hspace{1cm}
  \subfigure[]{\label{fig:firstex}\includegraphics[width=0.45 \textwidth]{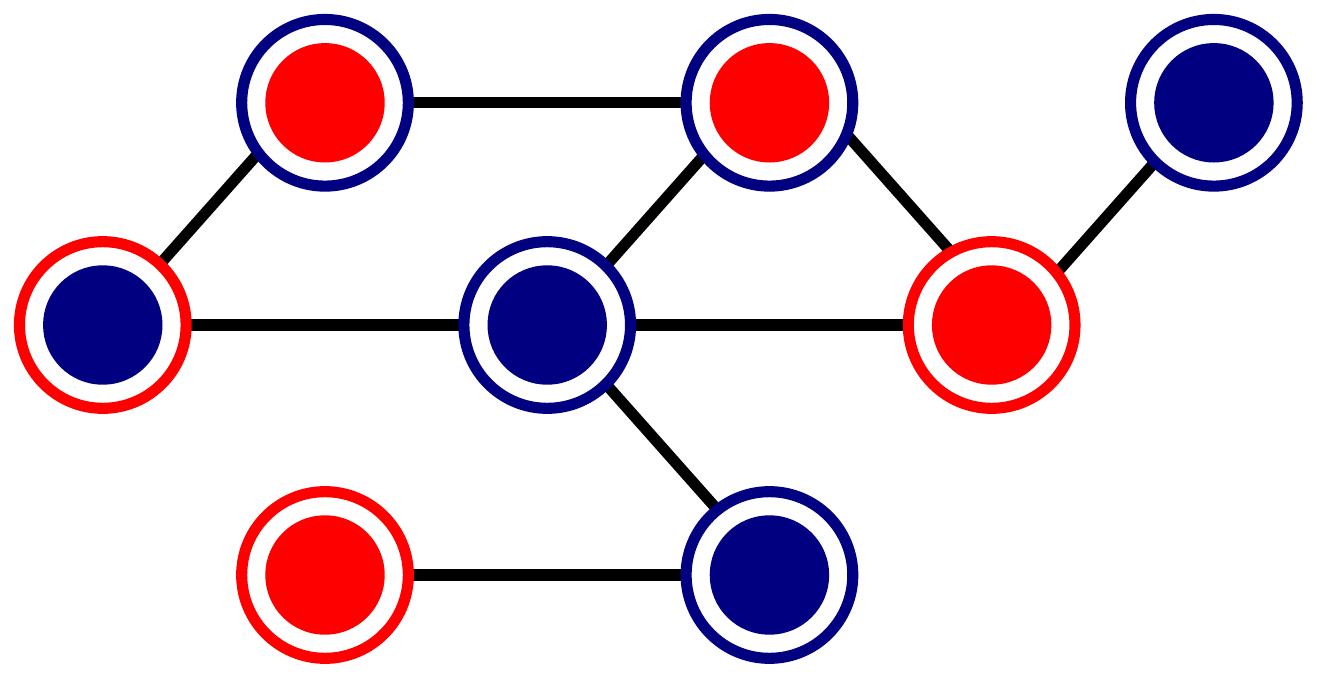}}
\caption{Examples of an environmental evolutionary graphs. The vertices are the thick circles and the individuals are the solid disks. If the colour of an individual $i$ matches the colour of the vertex $v_i$ then that individual is \emph{advantageous} at that vertex, and their fecundity is $f_i = r>1$. Otherwise, $f_i=1$.}
\label{fig:firstexample}
\end{figure}

Enviromental evolutionary graph theory is a fine-grained,
graph-theoretic analogue of a model first proposed in \cite{levene53},
where two niches are considered along with two alleles in the resident
diploid population, each advantageous in exactly one of the niches. It
was found that such a heterogeneous population can maintain a stable
genetic polymorphism even though the heterozygote is less fit than the
homozygote in its favourable niche. This model was later restricted to
haploid populations with multiple alleles at a single genetic locus,
each favoured in a different subset of sites in the environment;
this yields similar results on stable polymorphism
\cite{levinsmacarthur66, gliddonstrobeck75, strobeck79}. 

In the following, we define environmental evolutionary graphs and
prove some of their basic properties. We then extend the basic,
two-colour setup to multicoloured graphs. Doing so allows for
phenomena not present in the two-colour setup to emerge. For example,
the introduction of a third colour can permit a decrease in the time
to fixation of certain invading types. We then conclude with some
future prospects for research.

\section{Basic Properties}

Although the setup provided in the introduction is intuitive, we
require some formal definitions. Let $G$ be a graph on $N$ vertices
labeled $v_1, \ldots, v_N$. Each vertex has a \emph{background colour}
$c_i \in \{R,B\}$; if $c_i=R$, we think of vertex $v_i$ as red, and if
$c_i = B$, we think of vertex $v_i$ as blue. 

A \emph{state} of the model is a vector $(x_1, \ldots, x_N)$, where
each $x_i \in \{R,B\}$; the value of $x_i$ represents whether the
individual on vertex $v_i$ is currently red or currently blue. We call
$x_i$ the \emph{foreground colour} of vertex $v_i$ (in the given
state).  When the graph $G$ is understood, collection of all possible
states on $G$ is denoted $\mathbb{S}$. When all $x_i = R$, we are in
the \emph{all-red state}, which we denote $S_R$; similarly, when all
$x_i = B$ we are in the all-blue state $S_B$. When the process reaches
the all-red state, we say that the red type has achieved
\emph{fixation} in the graph, and likewise for blue. To avoid
ambiguity in phrases such as ``a red vertex'', we will capitalize
background colours; thus, ``a red vertex'' has red foreground colour,
while ``a RED vertex'' has red background colour.

The model has a single parameter $r$, which defines the reward
for an individual to match its color. The \emph{fecundity} of
individual $i$ in a given state is written $f_i$, and is defined by
$f_i = r$ if $x_i = c_i$ and $f_i = 1$ otherwise.

We define two possible transition rules between two-states, a
\emph{birth-death} rule and a \emph{death-birth} rule. These two rules
give rise to two different processes, the birth-death process and the
death-birth process. Both of these rules have been studied heavily in
the literature in the context of non-spatial Moran processes
\cite{moran58, nowaksasakitaylorfudenberg04}.

In a step of the birth-death process, we first choose an individual
$i$ reproduce; each individual is chosen with probability equal to
their relative fecundity, given by
\[
\mathbb{P}[\text{$i$ is chosen to reproduce}] = \frac{f_i}{\displaystyle \sum_{k\in V(G)} f_k}
\]
where the sum is taken over all vertices of the graph. Once an
individual is selected to give birth, it produces an offspring that
displaces a neighbour chosen uniformly at random (the offspring cannot
displace its parent). This assumption of uniform dispersal is not
necessary, and we later discuss properties of graphs that exhibit
biased dispersal.

In a step of the death-birth process, we instead start by choosing an
individual at uniform random to die. Neighbouring vertices then compete for
the vacated site according to their relative fecundities. Suppose
an individual $i$ dies. The probability that the neighbour $j$ places
an offspring on the vacant site is
\[
\mathbb{P}[\text{$j$ replaces $i$}] = \frac{f_j}{\displaystyle \sum_{k \in \mathcal{N}(v_i)} f_k}
\]
where the sum is taken over the set $\mathcal{N}(v_i)$ of all vertices adjacent to $v_i$.

Finally, given some initial state $S$ (and with the choice of
transition rule understood), we write $\rho_{R|S}$ to denote the
probability that the red type achieves fixation starting from
the state $S$. Often, we consider a single red mutant arising
at a uniformly selected vertex in an otherwise-blue graph;
the probability that red achieves fixation starting from this
initial distribution is simply written $\rho_R$. 


\subsection{Well-mixed Populations}

A natural first question to ask is, what is the effect of the density
$d$ of RED sites on the fixation probability $\rho_R$ of a red mutant?
To answer this, we first focus on a \emph{complete graph}, where all
pairs of distinct vertices are connected by an edge. This is an
example of a \emph{well-mixed} population. The following theorem
establishes that lowering the density of a type $X$ of sites lowers
the fixation probability of a set of $X$ types in a population
consisting of only two types undergoing a birth-death process.

To illustrate this, consider a well-mixed population $G$ of size $N$ undergoing a birth-death process with density $d$ of RED sites and suppose the fecundity $f_i$ of an individual $i$ that matches the type of their vertex $v_i$ is $f_i=r>1$ and $f_i=1$ otherwise. Using a mean-field approximation (see Appendix), we arrive at an equation for the fixation probability of a set of $m$ randomly placed red types, 
\begin{eqnarray}
\label{eq:fixprob}
\rho_{R|m}= \dfrac{1-\left( \dfrac{r(1-d)+d)}{(1-d)+rd}\right )^m }{1-\left(\dfrac{r(1-d)+d)}{(1-d)+rd}\right)^N}.
\end{eqnarray}

\begin{figure}
 \centering
 \includegraphics[width=0.75\textwidth]{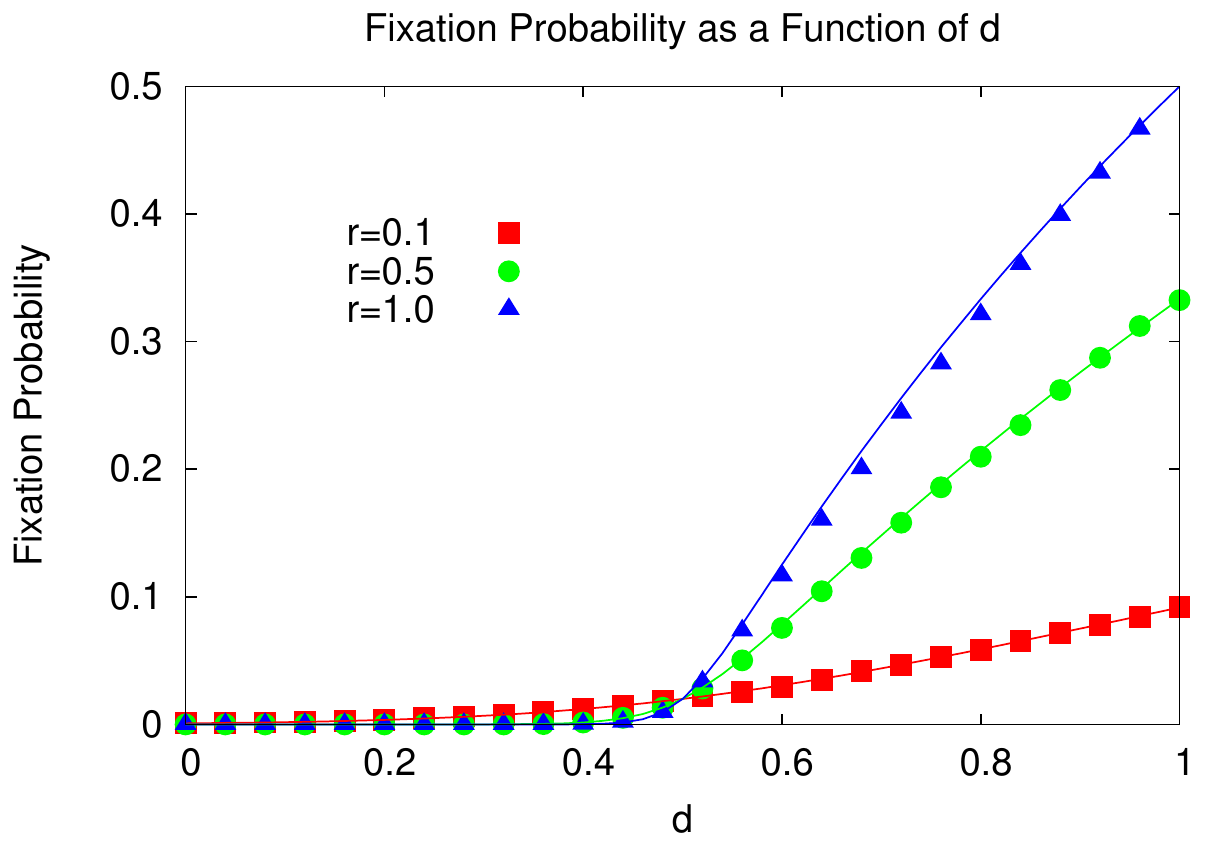}
 \caption{The fixation probability of a single red type on a complete graph with a fraction $d$ of RED sites. The points were generated with a simulation with results averaged over $10^6$ iterations. The solid curves were generated with Equation \ref{eq:fixprob}. In each of the three cases, a population of size $N=50$ was used. This choice of $N$ was arbitrary and we note that Equation \ref{eq:fixprob} was in good agreement with the simulation for various values of $N$ and $r$.}
 \label{fig:approx}
 \end{figure}

This approximation establishes that $\rho_{R|m}$ behaves as expected: the fixation probability of a set of $m$ red individuals increases as the density $d$ of RED sites increases, or as $m$ increases, for any fixed $r>1$. If less than half of the sites are RED, the fixation probability decreases in $r$ and if the density is greater than $1/2$ the fixation probability increases in $r$. 

A particular case of Equation \ref{eq:fixprob} of interest is for a single $R$ type. In this case, $m=1$ and Equation (\ref{eq:fixprob}) is
\begin{eqnarray}
\label{eq:fixprob1}
\rho_R = \dfrac{1-\left( \dfrac{r(1-d)+d)}{(1-d)+rd}\right ) }{1-\left(\dfrac{r(1-d)+d)}{(1-d)+rd}\right)^N}.
\end{eqnarray}
Since the derivative of Equation (\ref{eq:fixprob1}) with respect to $d$ is positive, it is an increasing function of $d$. Also, since the maximum value of Equation (\ref{eq:fixprob1}) is attained at $d=1$, then Equation (\ref{eq:fixprob1}) is strictly less than Equation (\ref{eq:fixprob1}) for all $0\leq d < 1$. That is, the fixation probability of a single $R$ type is lower if it is advantageous only on a proportion $d<1$ of sites than what it would be if it were advantageous everywhere, $d=1$.

It is worth noting that
\begin{eqnarray}
\lim_{d\to \frac{1}{2}} \rho_R = \frac{1}{N},
\end{eqnarray}
as is expected. That is, if half of the vertices of $G$ are RED and half are BLUE, then the rare red type fixes in the population as it would in a neutral population. However, this observation is valid only when considering the average fixation probability over all vertices. Where the rare type emerges may have a bearing on its fixation probability. 

\begin{figure}
 \centering
 \includegraphics[width=0.25\textwidth]{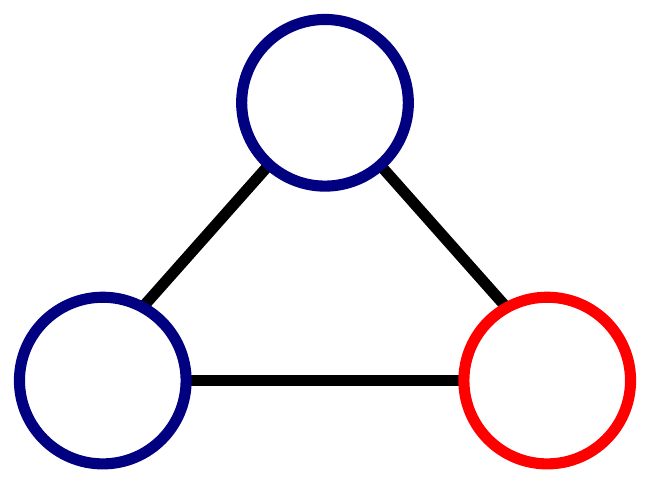}
 \caption{The $3$-cycle with one RED and two BLUE vertices.}
 \label{fig:3cycle}
 \end{figure}

As an example, consider a cycle graph on three vertices. Colour two of the vertices BLUE and the other RED, as in Figure \ref{fig:3cycle}, and suppose the population is undergoing a birth-death process. If a red individual appears on one of the BLUE vertices, it has fixation probability
\begin{eqnarray}
 \rho_{R|B} = \frac{3r^2+7r+5}{7r^2+19r+19}. 
\label{eq:cyclebluefixprob}
\end{eqnarray}
If it appears on the RED vertex, its fixation probability is
\begin{eqnarray}
 \rho_{R|R} = \frac{r(7r+8)}{7r^2+19r+19}.
\label{eq:cycleredfixprob}
\end{eqnarray}
A quick comparison of Equations (\ref{eq:cyclebluefixprob}) and (\ref{eq:cycleredfixprob}) indicates $\rho_{R|R}>\rho_{R|B}$ for all $r>1$. Hence, the fixation probability $\rho_R$, in general, depends on the starting location of $R$. 

This leads to an interesting question: for what graph colourings is $\rho_R$ independent of the starting position of the single red individual? This is answered in the next section. 

\section{Properly Two-coloured Graphs}

In this section we focus on a specific type of colouring of a environmental evolutionary graph: \emph{proper two-colourings}. We will suppose that the edges carry the uniform weighting: $w_{ij} = 1/d_i$ for all $i$ and $j$ adjacent vertices.
\begin{definition}
A properly two-coloured graph is one with no two adjacent vertices coloured the same.
\end{definition}
The lattice in Figure \ref{fig:latticeex} is an example of a properly two-coloured graph. Figure \ref{fig:twocoloured} provides two more examples. A properly two-coloured graph can have any number of vertices of any number of degrees provided that the vertices of the graph are coloured so that no two vertices of the same colour are adjacent. The class of graphs that can be properly two-coloured are known as the \emph{bipartite} graphs. Such graphs are an active topic of study; see \cite{diestel10} for a thorough introduction to bipartite graphs.

\begin{figure}

\centering
  \subfigure[]{\label{fig:starex}\includegraphics[ width=0.35\textwidth]{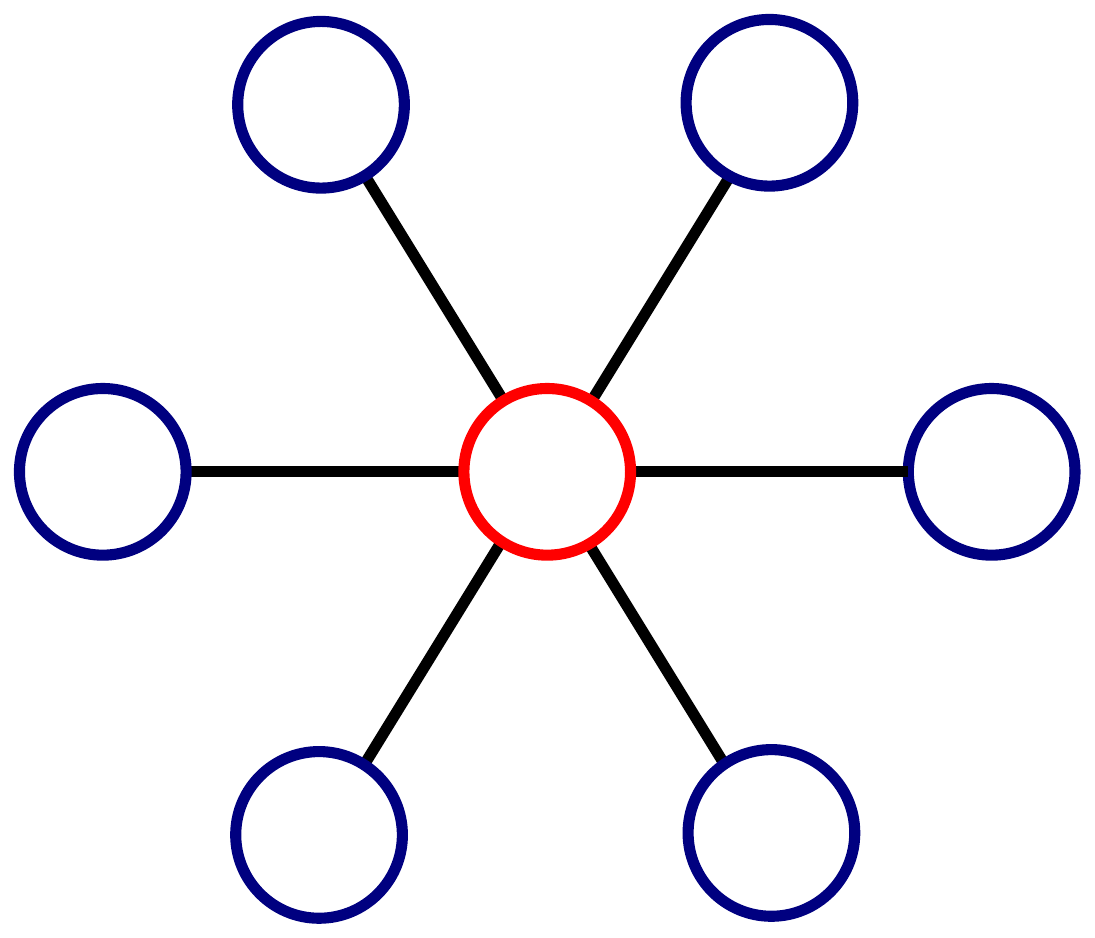}}
 \hspace{1cm}
  \subfigure[]{\label{fig:6cycle}\includegraphics[width=0.35 \textwidth]{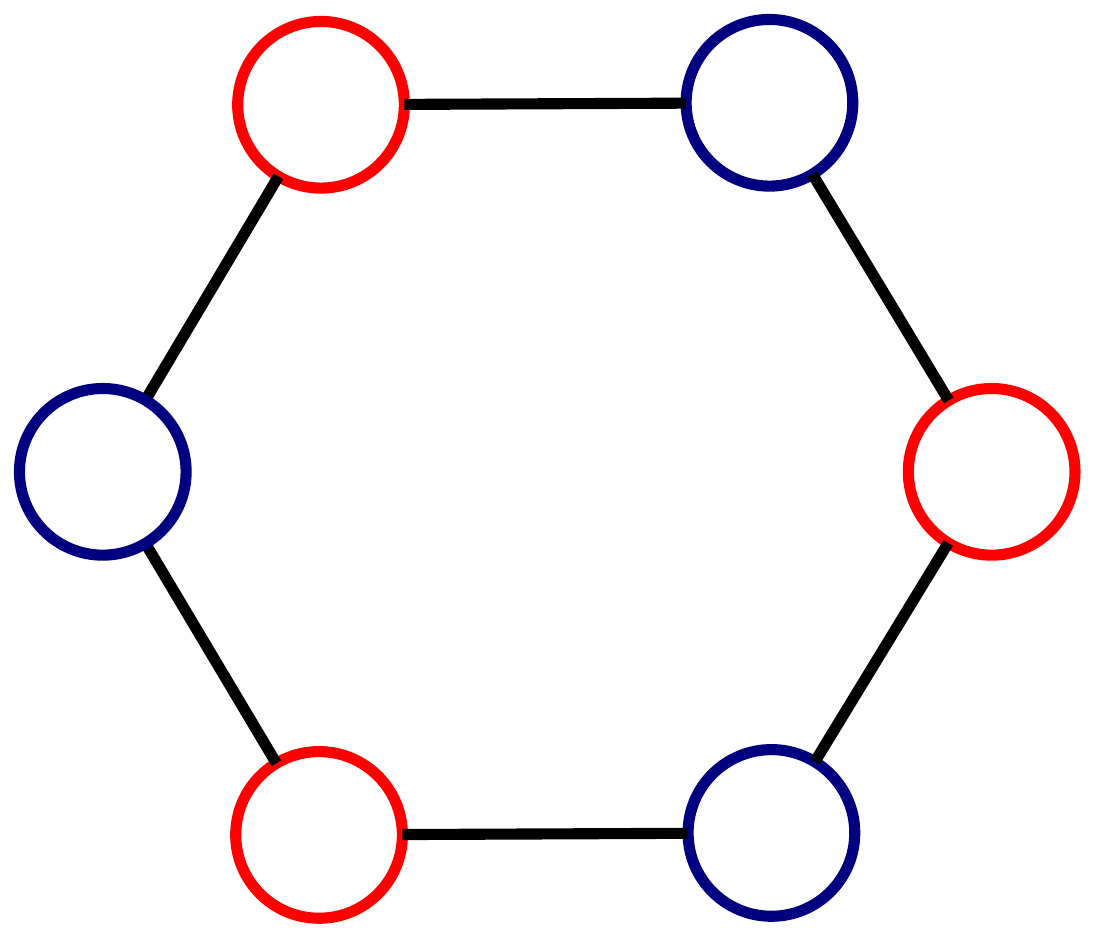}}
\caption{Two examples of properly two-coloured graphs. }
\label{fig:twocoloured}
\end{figure}

Properly two-coloured graphs exhibit a fascinating property: the birth-death evolutionary process on such graphs does not depend on the parameter $r$. More precisely, the fixation probability of a set of red types on a properly two-coloured graph is equal to the corresponding neutral fixation probability. Recall that a neutral process is one which $r=1$.  In the context of environmental evolutionary graphs, the population is neutral if both red and blue types have fecundity $1$, irrespective of the vertices they occupy. We have the following.
\begin{theorem}
Given a properly two-coloured graph $G$ undergoing a birth-death process and a set $M\subset V(G)$ of vertices occupied by $R$ (red) types, then the probability $\rho_{R|M}$ that the $R$ fix in the population is
\begin{eqnarray}
\label{eq:twocoloured}
\rho_{R|M} = \sum_{i \in M} \rho_{neutral|i},
\end{eqnarray}
where $\rho_{neutral|i}$ is the neutral fixation probability of an individual starting at vertex $v_i$. 
\label{thm:twocoloured}
\end{theorem}

\begin{proof}
The proof of this theorem requires some technical results and is left to the appendix.
\end{proof}

Equation (\ref{eq:twocoloured}) has a convenient form in terms of \emph{reproductive value}. Recall that the reproductive value of an individual $i$ is the (relative) probability that a member of the population at some time in the distant future is identical by descent to $i$ \cite{fisher30, taylor90, grafen06}. In \cite{maciejewski13} the reproductive value $V_i$ was calculated for any vertex $v_i$ in any evolutionary graph undergoing either a birth-death or death-birth process. It was found that $V_i = d_i$ for the death-birth process and $V_i = 1/d_i$ for the birth-death process. Moreover, the author of \cite{maciejewski13} shows that the neutral fixation probability of a mutant starting on vertex $v_i$ is equal to $i$'s relative reproductive value:
\begin{eqnarray}
\rho_{neutral|i} = \dfrac{V_i}{\sum_{j \in V(G)}V_j}.
\end{eqnarray}
Hence, in terms of vertex degrees, Equation (\ref{eq:twocoloured}) reads
\begin{eqnarray}
\rho_{R|M} = \dfrac{\displaystyle \sum_{i \in M} \dfrac{1}{d_i}}{\displaystyle \sum_{j \in V(G)} \dfrac{1}{d_j}},
\label{eq:rv}
\end{eqnarray}
which is exactly what is expected given the results of \cite{maciejewski13}. It is worth emphasizing that Equation (\ref{eq:twocoloured}) does not depend on $r$, a fact made transparent by Equation (\ref{eq:rv}).

It is interesting to note that Theorem \ref{thm:twocoloured} does not hold in general for the death-birth process. In fact, counter-examples are easy to come by. Take, for example, the section of a line graph, as in Figure \ref{fig:counterexample}. The graph is properly two-coloured yet the fixation probability is not independent of $r$. 


%

\begin{figure}
\centering
\includegraphics[width=0.6\textwidth]{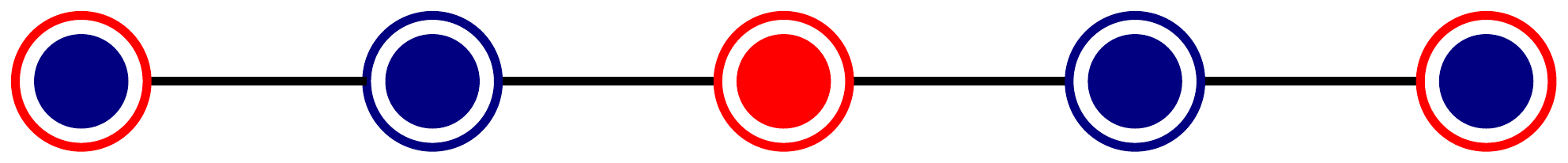}
\caption{An example of a properly two-coloured graph on which the fixation probability of a red type in the death-birth process depends on $r$.}
\label{fig:counterexample}
\end{figure}

We conjecture that there is no class of environmental evolutionary graphs on which the fixation probability in the death-birth process is independent of $r$. This is perhaps surprising since the author of \cite{maciejewski13} was able to show that for \emph{regular}---meaning all vertices have the same degree---evolutionary (non-environmental) graphs undergoing a death-birth process, and for a set $M$ of $R$ types, 
\begin{eqnarray}
\rho_{R|M} = \dfrac{\displaystyle \sum_{i \in M} d_i}{\displaystyle \sum_{j \in V(G)} d_j}.
\end{eqnarray}
Such a result suggests an extension to environmental evolutionary graphs as was the case for the birth-death process. The reason why the results of \cite{maciejewski13} can be extended to environmental evolutionary graphs for the birth-death process and not for the death-birth process is the scale of information about the population required by the two processes. The death-birth process requires very \emph{local} information about the population state, namely, the state of the neighbours of an individual chosen to die. The birth-death process requires \emph{global} information about the population; it requires the state of all individuals in the population. This global property allows for a complete classification of graphs on which the fixation probability in the birth-death process is independent of $r$. The local nature of the death-birth process imposes different local conditions for the independence of the fixation probability on $r$, which may conflict and not scale to the entire population.

\subsection{On the Starting Location}

The structure of the underlying graph $G$ and its associated colouring can affect the result of the model in sometimes counterintuitive ways. As an example, it is natural to assume that the local fitness advantage conferred by matching the background colour of a vertex necessarily translates into a global fitness advantage for a lone mutant starting at that vertex. This notion is formalized in the following conjecture:
\begin{conjecture}
  Let $G$ be an environmental evolutionary graph and let $v,w$ be vertices in $G$. Let $\rho_{R|v}$ denote the probability that red achieves fixation starting from the state where $v$ is the only red individual, and likewise for $\rho_{R|w}$. If $v$ is coloured RED and $w$ BLUE, then $\rho_{R|v} \geq \rho_{R|w}$.
\end{conjecture}
This natural conjecture turns out to be false, as it fails to take into account the global structural characteristics of the graph. Indeed, we present a counterexample in which the underlying graph $G$ is symmetric and is equally suited for red and blue, yet a red that emerges on a certain BLUE vertex experiences a fixation probability greater than if it had emerged on a corresponding RED vertex. This example illustrates that an initial fitness disadvantage can be offset by a subsequent fitness advantage. 

Our counterexample is a weighted graph and therefore uses the weighted version of the model, so that if a vertex $u$ is selected to reproduce, then the probability that its neighbor $v$ is selected to die is proportional to the weight of the edge $e_{uv}$. The graph consists of two triangular clusters $T_1$ and $T_2$ whose edges are uniformly weighted. All possible edges between $T_1$ and $T_2$ are included with weight $\alpha$, where $\alpha$ is a small constant to be determined later. In $T_1$ there are two BLUE vertices and one RED vertex, while in $T_2$ there are two RED vertices and one BLUE vertex. The graph and its colouring are illustrated in Figure \ref{fig:islandex}.
\begin{figure}
\centering
\includegraphics[width=0.6\textwidth]{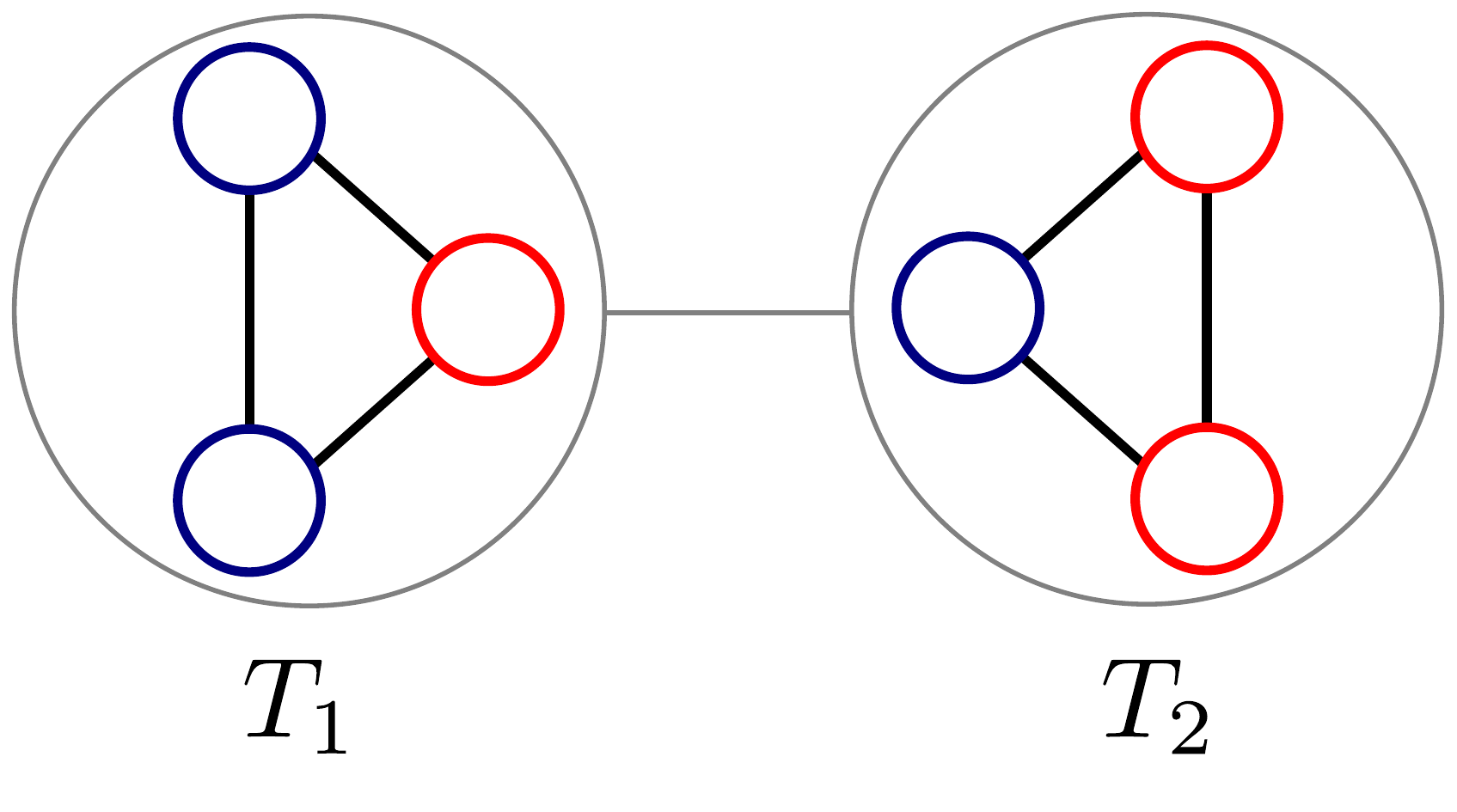}
\caption{By emerging on a BLUE vertex, a red may experience a fixation probability greater than if it had emerged on a RED vertex even though both red and blue are equally well-suited to the environment.}
\label{fig:islandex}
\end{figure}
Clearly $G$ is symmetric, even when the edge-weights are taken into account.  Assume that $r$ is much greater than $1$. Let $v$ denote the unique RED vertex in $T_1$ and let $w$ denote the unique BLUE vertex in $T_2$. We will argue that $\rho_{R|w} > \rho_{R|v}$, i.e., a single red mutant has a better chance of achieving fixation if it arises on the BLUE vertex $w$ than it does if it arises on the RED vertex $v$. We first give a heuristic, non-rigorous argument that nevertheless expresses \emph{why} this result should be the case. Then, we give specific parameter values and obtain the relevant fixation probabilities numerically, which further establishes the result.

Heuristically, the effect of taking $\alpha$ ``sufficiently small'' is that the edges joining $T_1$ and $T_2$ are almost never selected. Thus, given any initial state, almost surely the triangles $T_1$ and $T_2$ will fixate on a particular population colour before any cross-edge is selected for reproduction. If $T_1$ and $T_2$ fixate on the same colour, then $G$ has achieved fixation. If $T_1$ and $T_2$ fixate on the opposite colours, then by obvious symmetry the probability that $G$ fixates on red is $1/2$. It therefore suffices to consider the probability that $T_1$ fixates on red when a red mutant arises on $v$, which we denote $\rho_1$, as well as the corresponding probability for $T_2$ and $w$, which we denote $\rho_2$.

First we estimate $\rho_1$. Observe that if the red mutant arises on $v$, then initially all individuals in $T_1$ match their vertex colour and therefore have fitness $r$; thus, on the first step all individuals are equally likely to be chosen. With probability $1/3$ nothing changes (a blue individual is chosen but replaces the other blue individual), with probability $1/3$ the red mutant is immediately replaced with blue, and with probability $1/3$ the red mutant is chosen to reproduce, replacing a blue individual. Thus, conditioning on the event that the state of the graph changes, with probability $1/2$ we end up with two red vertices. Since there is still non-negligible probability that the remaining blue individual will be chosen to reproduce and overwrite the red individual on $v$, this implies that $\rho_1 < 1/2$. Furthermore, as $\alpha \to 0$ and $r \to \infty$, the probability $\rho_1$ will be bounded away from $1/2$.

Next we estimate $\rho_2$. If the red mutant arises on $w$, then the opposite situation reigns: all individuals in $T_2$ initially fail to match their vertex colour and have fecundity $1$. Thus, all individuals are equally likely to be chosen, but as soon as an individual manages to occupy a vertex of the correct colour, the process will almost surely fixate on that colour. This implies that $\rho_2 \approx 1/2$ since, as in the earlier analysis, the probability that $w$ reproduces, conditional on a change in state, is about $1/2$. We heuristically conclude that $\rho_1 < \rho_2$.

Next we examine the situation numerically to confirm our heuristic analysis. Using the graph described above, let $r = 1/4$ and let $\alpha = 1/100$. Using standard techniques of Markov chain theory, we numerically compute that $\rho_{R|v} \approx 0.13$ and $\rho_{R|w} \approx 0.2$, so that $\rho_{R|w} > \rho_{R|v}$ as desired.

\section{More Than Two Background Colours}

We now introduce a third colour, green, for both vertices and individuals. We retain the previous notion of fecundity, that if an individual's colour matches that of the vertex then its fecundity is $f=r>1$ and $f=1$ otherwise. For ease of presentation, attention is limited to the birth-death process for this entire section. 

Introducing a third colour can never increase the average fixation probability of any single colour of individuals. However, and quite interestingly, a third colour can decrease the time to fixation of a mutant individual. 

As was seen in Section 2 for well-mixed populations, via Equation \ref{thm:fixprob}, lowering the density of sites of a certain colour decreases the average fixation probability of the individuals of that colour. This can also seen to be true for graph-structured populations undergoing either the birth-death or death-birth process: any non-RED vertex will eventually be occupied by a red individual and this individual is less fit than it would have been if the vertex were RED. Such an effect also occurs on graphs with more than two vertex colours.  

As an example, consider the line graph on three vertices, coloured as in Figure \ref{fig:onered}, and suppose the population is undergoing a birth-death process. Suppose further that the population initially consists entirely of either of blue or green until a mutation occurs, producing a red. This mutant appears on the hub vertex with probability $(1+r)/(2+r)$ and on one of the leaf vertices with probability $1/(2+r)$. A simple calculation of the average fixation probability yields
\begin{eqnarray}
\overline{\rho}_R = \frac{2r^2(r+1)}{2r^3+5r^2+4r+4}.
\label{eq:onered}
\end{eqnarray}
This is seen to be less than the corresponding fixation probability in an all RED environment,
\begin{eqnarray}
\overline{\rho}_R = \left(\frac{1}{5}\right)\frac{12r^2+7r+1}{6r^2+7r+2}.
\label{eq:allred}
\end{eqnarray}
which is illustrated in Figure \ref{fig:fixprobcomp}. 



\begin{figure}[h]
\centering
  \subfigure[]{\label{fig:onered}\includegraphics[ width=0.05\textwidth]{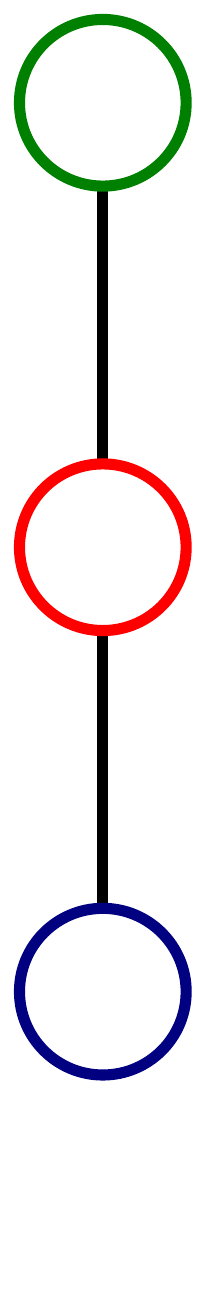}}
 \hspace{2cm}
 \label{fig:fixprobcomp}
  \subfigure[]{\label{fig:r}\includegraphics[width=0.65\textwidth]{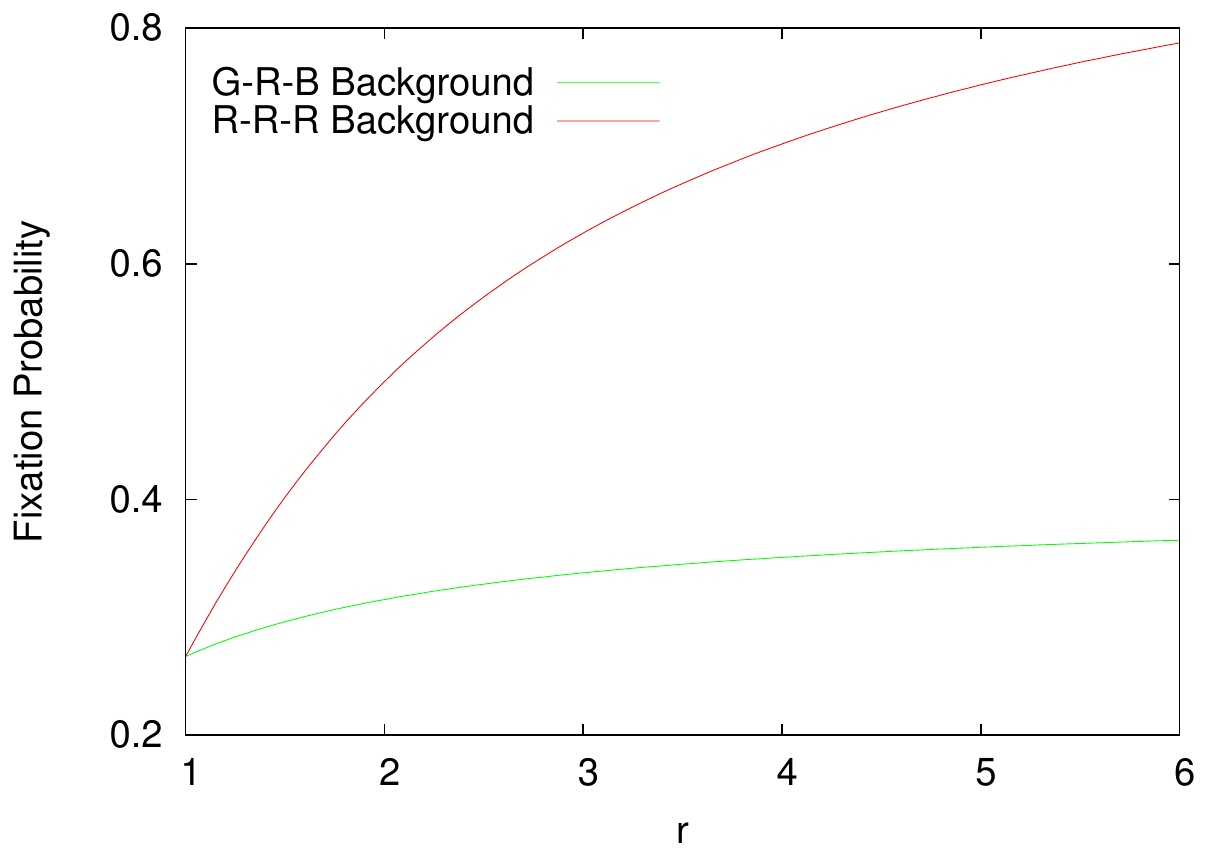}}
\caption{A third colour may decrease the fixation probability of a red mutant. (a) The line graph on three vertices with the centre coloured RED and the leaves GREEN and BLUE. (b) The introduction of the third vertex colour, GREEN, decreases the fixation probability of red by lowering the proportion of RED sites. The solid line corresponds to the average fixation probability of a red mutant on a $3$-line graph consisting entirely of RED vertices. The dotted line is the average fixation probability of a red mutant on a $3$-line graph with hub coloured RED, as in Figure \ref{fig:onered}. Both curves are functions of $r$.} 
\end{figure}

Supposing that the red type does go on to fixation, we may determine the time such an event takes. The expected number of birth-death events needed for the red type to fix in the population is known as the time to fixation \cite{ewens04, antalscheuring06, traulsenhauert09}. 

Figure \ref{fig:fixationtimeline} displays the time to fixation for a mutant red type on both the environmental graph in Figure \ref{fig:onered} and in an all-RED $3$-line. In both cases, the mutant red appears in a population of either all blue or green individuals. The calculations for the times to fixation are in the Appendix. 

The average time to fixation of a single red is seen to be lower in the population depicted in Figure \ref{fig:onered} than that for an all-RED $3$-line for a range of $r$ values. This difference in fixation time is easily explained by considering the population update rule. 

For the birth-death process, those with greater fecundity are chosen more often to reproduce. If the RED site is located on the hub vertex then a mutant red type on this vertex has an advantage over the resident blue (green) type on a GREEN (BLUE) leaf; it will be chosen with greater probability than such a blue (green) type. Once it is chosen for reproduction it places a red offspring on one of the leaves. Suppose it displaces a leaf individual that does not match their vertex colour. This new red offspring has the same fecundity as the individual it replaced. Hence, the red type on the hub maintains its fecundity advantage. If the environment had a RED leaf where the red offspring was placed then this new offspring would also have a fecundity advantage and would be more likely to compete with the red type on the hub for reproduction. This would result in the leaf red displacing its parent more often. This reproductive event is ``wasted'' in a sense, since it did nothing to bring the population closer to an 
all-red state. This type of redundant back-and-forth is reduced if the red type does not experience an increase in fecundity on the leaf vertices.


\begin{figure}[h!]
\centering
\includegraphics[width=0.8\textwidth]{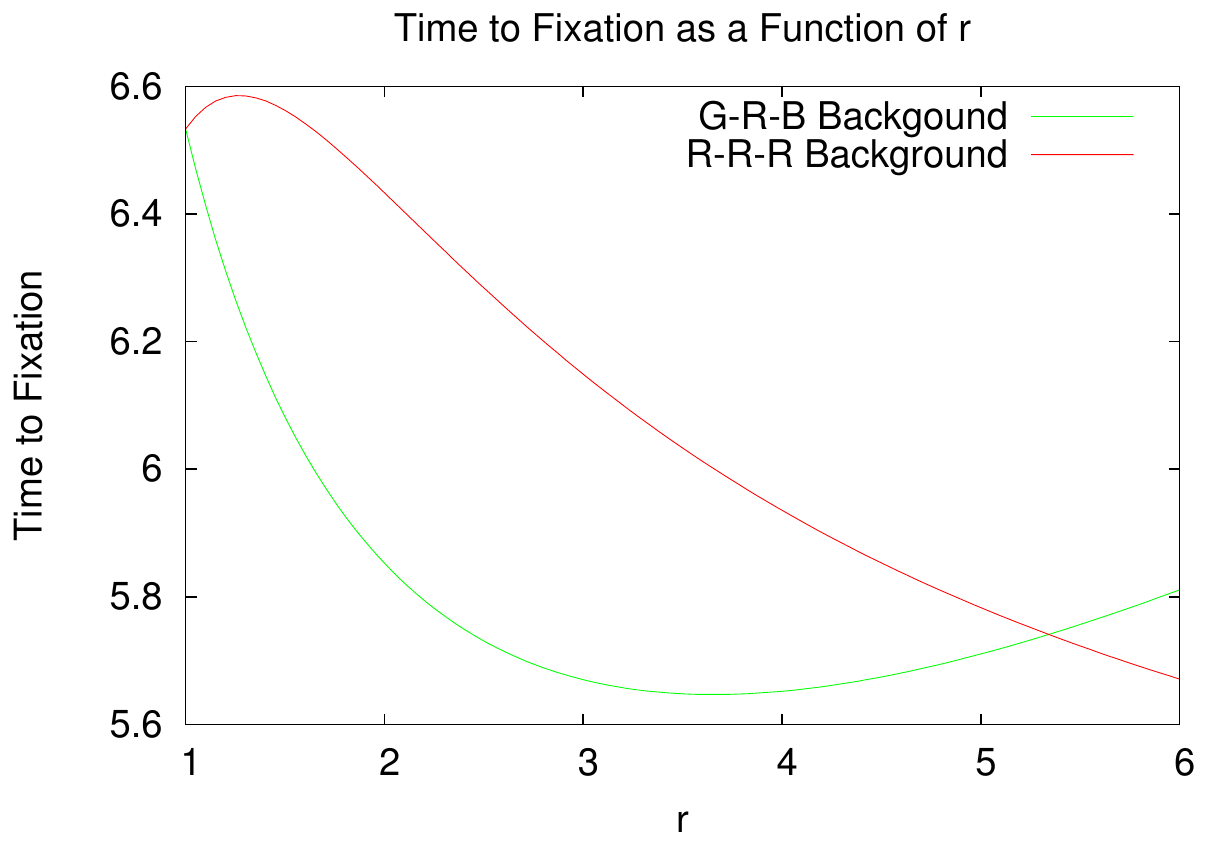}
\caption{Lower densities of RED sites may decrease the time to fixation. The curves correspond to the expected number of steps required for a single red individual to reach fixation in a population otherwise composed of all blue or all green. The dotted line corresponds to the single RED environment in Figure \ref{fig:onered}, denoted $G-R-B$.  The solid line corresponds to the all-RED $3$-line, denoted $R-R-R$. The average time to fixation of a single red in the $G-R-B$ environment is less than than that in an all-RED environment. This is especially pronounced for values of $r$ less than approximately $3.5$. The time to fixation in the $G-R-B$ environment eventually increases and surpasses that of the all-RED environment because the average is taken over all possible starting positions for the red mutant. If this red type emerges on the BLUE vertex then, for $r$ sufficiently large, it will be chosen to reproduce with very low probability. Therefore, the time to fixation gets increasingly large. As $r$ 
increases, so does this fixation probability.}
\label{fig:fixationtimeline}
\end{figure}

This simple example illustrates a more general observation: an advantageous mutant decreases its time to fixation in a population by not interfering with copies of itself. This phenomenon of decreasing time to fixation appears to not be restricted to this toy example. Figure \ref{fig:BAresults} presents the average time to fixation for a red type on three different random graphs generated with a preferential attachment algorithm. These graphs are known to approximate authentic social interaction networks \cite{barabasialbert99}. For each randomly-generated graph, a minimum average time to fixation is obtained for some intermediate proportion of vertices coloured RED. The upshot of these observations is that the introduction of a third colour can decrease the average time to fixation of a mutant type. 


\begin{figure}[h!]
  \centering
  \subfigure[]{\label{fig:BA}\includegraphics[width=0.3\textwidth]{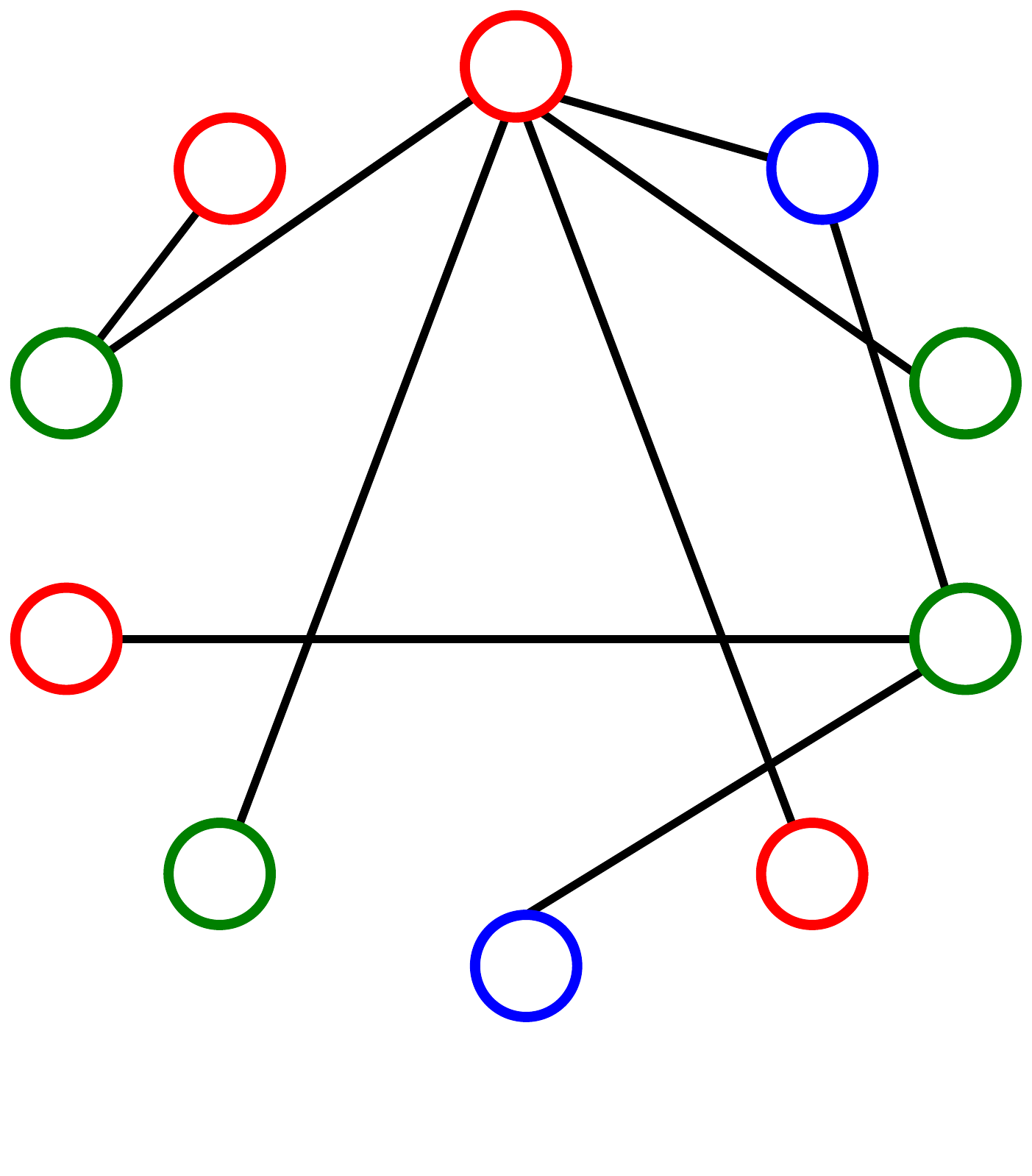}}
\hspace{1cm}
  \subfigure[]{\label{fig:BAtimes}\includegraphics[width=0.6\textwidth]{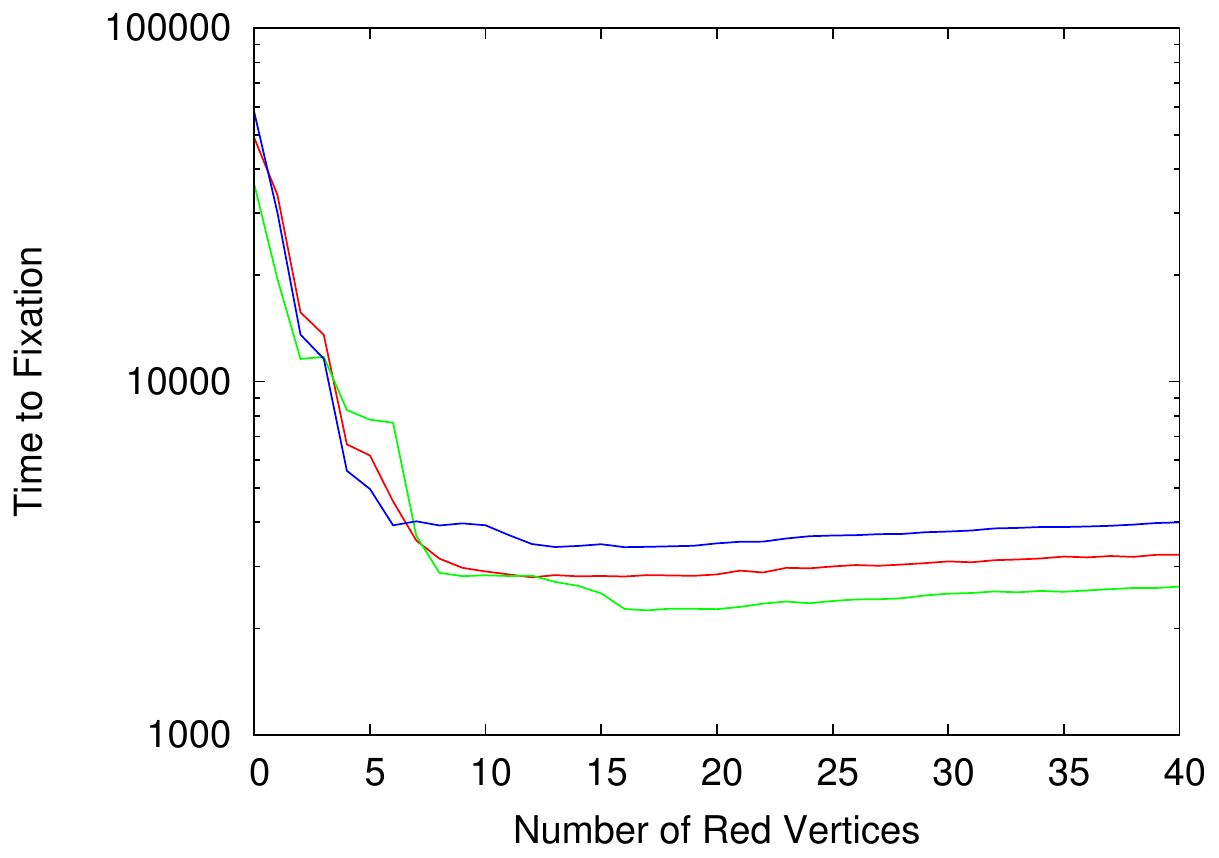}}
\caption{Barabasi-Albert scale-free graphs are a type of random graph generated with a preferential attachment algorithm \cite{barabasialbert99}. These graphs exhibit a power law degree distribution: very few vertices are of a higher degree, while many are at a low degree. An example graph generated with a Barabasi-Albert algorithm is in (a). Figure \ref{fig:BAtimes} reports the time to fixation in three sample BA graphs, each of size $N=40$ and the fecundity difference set arbitrarily at $r=2$. Once one of the BA graphs is generated, the vertices are ranked according to degree. The top $m$ vertices are coloured RED and the evolutionary process begins with a single red individual placed randomly on the graph. The process is run until $10,000$ fixation events occur.  The average number of steps required to reach fixation is plotted, using a log scale, against the number $m$ of RED sites. The key observation is that there is an initial steep drop in fixation time until the time is minimized. A gradual increase 
in time is observed as more of the vertices are coloured RED. In each example, the time to fixation when all the vertices are coloured RED is roughly $30\%$ greater than the minimum time to fixation, which generally occurs when roughly $16$ of the $40$ vertices are RED. }
\label{fig:BAresults}
\end{figure}

\subsection{Measures of Advantage}

Throughout this manuscript, we have assumed that a single type $X$ emerges in a population of all-$Y$ and either goes on to fixation or dies out. This notion of a single new type emerging in a pure state is a result of assuming that the probability of mutation $\mu$ is so small that the time between mutation events is much larger than the time it takes for a mutant individual to fix in, or die out from, a population. If we were to observe a population undergoing such a mutation/fixation/extinction process at a given point in time, then with high probability it will be in a pure state. Equivalently, as time goes to infinity the proportion of time the population spends in a pure state approaches $1$. This observation can be used as a measure of evolutionary advantage \cite{roussetbilliard00}: the more suited a type $X$ of individual is to the environment, the greater the expected time the population spends in state all-$X$.

Define $q_X$ to be the expected proportion of time the population spends in the all-$X$ state $S_X$, where $X$ is any permitted colour. This probability will, in general, depend on the mutation rates and the fixation probabilities. This is formalized in the following.

Again, since mutational events likely only occur when the population is in a pure state, we can consider the population as a Markov chain transitioning between pure states. This process is Markovian since the current state of the population only depends on the previous state. Having established this, we can write a general balance equation for the Markov chain:
\begin{eqnarray}
q_X \sum_{Y \in C\setminus \{X\}}\mu_{XY}\ \rho_{Y|X} = \sum_{Y \in C\setminus \{X\}} q_Y \ \mu_{YX} \ \rho_{X|Y},
\label{eq:balance}
\end{eqnarray}
where the sums are taken over all colours except $X$ and $\rho_{X|Y}$ is the probability that a single $X$ individual fixes in a population of all $Y$, and $\mu_{XY}$ is the probability a $Y$ appears in an all-$X$ population through mutation. The terms $\rho_{X|Y}$ and $\mu_{XY}$ in Equation (\ref{eq:balance}) are not simply the average fixation probability or mutation, but may depend on the configuration of $X$ and $Y$ sites and where in the population the $X$ or $Y$ emerges \cite{maciejewskihauertfu13}. This will be illustrated in a series of examples.

Equation (\ref{eq:balance}) establishes a system of equations for the $q_X$. A unique solution is found by incorporating the equation $\sum_{i} q_i = 1$. In general, the solution to this system is cumbersome, but a compact, intuitive solution can be given in certain situations.

A \emph{vertex-transitive} graph $G$ has the property that for any two vertices $v_i$ and $v_j$ of $G$, there exists an automorphism (a mapping from $G$ to $G$) $f$ of the vertices of $G$ such that $f(v_i) =v_j$, that is, $v_i$ is mapped into $v_j$ while preserving the structure of the graph. Intuitively, this property asserts that all vertices are equivalent; the graph ``looks" the same from any two vertices. Here we suppose that the symmetry is a property of the graph structure only. Due to the relative ease of calculations on vertex-transitive graphs, this class is extensively studied in the evolutionary graph theory literature \cite{taylordaywild07a, taylorlillicrapcownden11}. 

Properly two-coloured vertex-transitive graphs, like the $6$-cycle in Figure \ref{fig:6cycle} are a class of graphs on which Equation (\ref{eq:balance}) is easily solved. For the birth-death process Theorem \ref{thm:twocoloured} established that the fixation probability of either a red or blue type in a properly two-coloured graph is equal to the neutral fixation probability. Moreover, for properly two-coloured vertex-transitive graphs $\mu_{RB} = \mu_{BR}$, since for every instance of a red emerging on a colour $X$ vertex there is a corresponding instance of a blue emerging on an $X$ vertex with the same probability. All told, Equation (\ref{eq:balance}) reduces to $q_R = q_B=1/2$ on such graphs.


If we consider non-vertex-transitive graphs then the possibilities for the $q_X$ are many. Figure \ref{fig:threeRGBexamples} displays three graphs each with equal proportion of RED, GREEN, and BLUE sites, yet each example favours the three colours of individuals differently. In the $3$-cycle of Figure \ref{fig:3cycleRGB}, the fixation probabilities of each colour are equal as are the probabilities of any one type emerging in a pure state of any other type. Hence, $q_R = q_G=q_B = 1/3$. Figure \ref{fig:3lineRGB} is our $3$-line example from earlier. Supposing the probability of mutation between any two types is the same, we expect that the environment is ``most-suited'' for red---that is, red should have the greatest fixation probability---less-so for green and blue. We expect to find the population in a state of all-red more often than all-green and all-blue. In our notation, $q_R > q_G$ and $q_R > q_B$. Indeed, a quick calculation reveals that this is the case. In this example, $q_B=q_G$, but this does not 
necessarily follow from the advantage of red over blue and green. It is also possible to find a structure such that $C_R = C_G = C_B$, yet $q_R > q_G > q_B$. An example is Figure \ref{fig:weighted3lineRGB}. Here no edges emanate from the BLUE vertex so that any offspring produced there fail to secure a site. 

\begin{figure}[h!]
  \centering
  \subfigure[]{\label{fig:3cycleRGB}\includegraphics[width=0.25\textwidth]{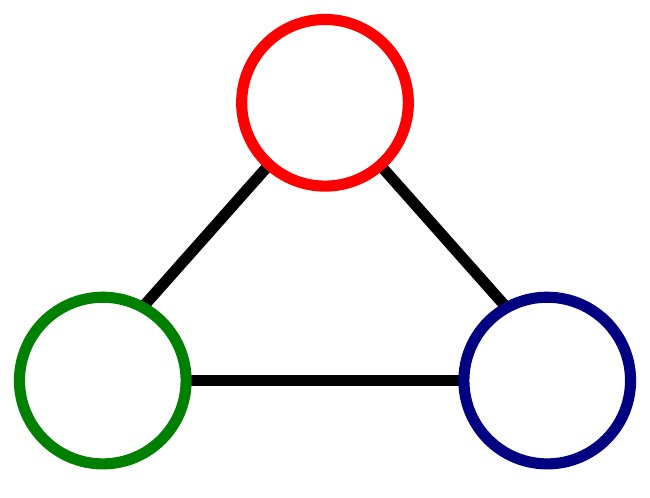}}
\hspace{1cm}
  \subfigure[]{\label{fig:3lineRGB}\includegraphics[width=0.25\textwidth]{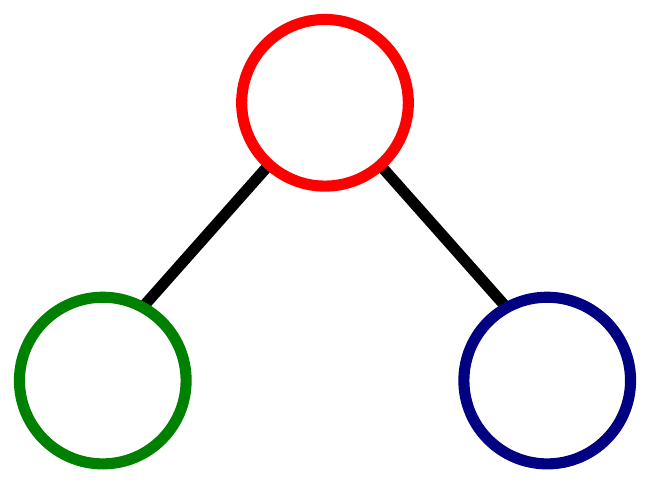}}
\hspace{1cm}
  \subfigure[]{\label{fig:weighted3lineRGB}\includegraphics[width=0.25\textwidth]{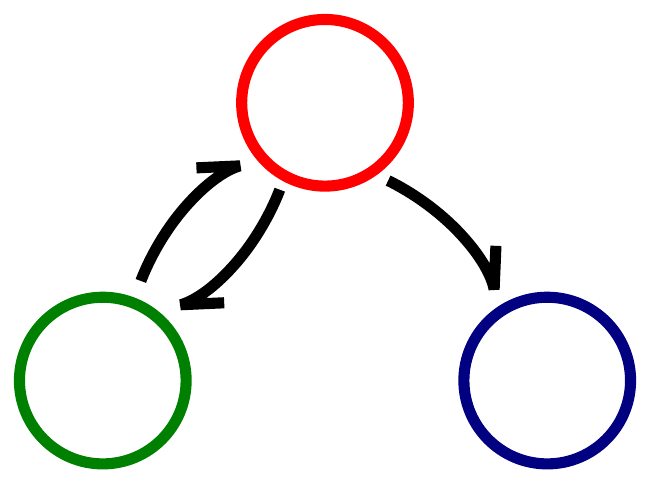}}
 \caption{Three examples of graphs with equal proportions of RED, GREEN, and BLUE sites that all differ in the evolutionary advantage experienced by the three corresponding colours of individuals. Edges indicate where offspring disperse. In a) and b) the edges are weighted uniformly. So, any offspring produced on any vertex in a) will disperse to a neighbouring vertex with probability $1/2$, while in b) offspring produced on the RED vertex will disperse to either neighbouring vertex with $1/2$, while an offspring produced on the GREEN or BLUE vertices will disperse to the RED vertex with probability $1$. For c), the edges are weighted as in b) yet there is no edge eminating from the BLUE vertex. Any offspring produced on the BLUE vertex does not disperse and fails to secure a vertex.  Denoting the proportion of time spent by a population in a state of all-$X$ by $q_X$, then, a) $q_R=q_B=q_G$; b) $q_R > q_G, q_B$, and $q_G=q_B$; and c) $q_R>q_G>q_B$. }
\label{fig:threeRGBexamples}
\end{figure}

The key observation here is, in general the fraction of the habitat best suited for a type $X$ is not sufficient to determine the evolutionary advantage of $X$. Information on the spatial arrangement of the sites is also important. This result is similar to results in the ecology literature \cite{cantrellcosner91, clarkethomaselmeshochberg97, latoregouldmortimer99, hiebelermichaudwassermanbuchak13}. For example, the authors of \cite{cantrellcosner91} consider a continuous, one-dimensional environment and suppose that some of the regions in the environment are more favourable than others. They found that it was not the proportion of these regions that mattered for population persistence, rather their location within the environment. Or, using a integro-difference equations, the authors of \cite{latoregouldmortimer99} found that an environment with lower-quality regions distributed throughout may be more suitable for a population than an environment of uniform high quality. 

\section{Discussion}

Building an understanding of how an environment shapes the evolution of a population is an ongoing challenge. There are now a plethora of models describing the evolution of populations in structured environments. These include island and deme structured \cite{wright31}, stepping stone \cite{kimuraweiss64}, lattice \cite{nowakmay92, nakamarumatsudaiwasa97}, metapopulations \cite{levins69}, and evolutionary graphs \cite{liebermanhauertnowak05}. Our model extends these spatial models by incorporating location-dependent fecundity. There is considerable evidence that patch quality affects an evolutionary process \cite{hobbshanley90, fleishman02, yamanaka09}. Our model allows explication of the effects of population structure and patch quality on an evolutionary process. 

The direct precursor to our model, evolutionary graph theory, is an extremely active area of research; see \cite{szabofath07} for a review. The constant fecundity process, as introduced in \cite{liebermanhauertnowak05}, is very well-understood---the \emph{circulation theorem} of \cite{liebermanhauertnowak05} completely describes the process on a large class of graphs. However, the majority of results in the evolutionary graph theory literature rely on some sort of symmetry in the population \cite{ohtsukinowak06, taylordaywild07a}. The challenge is to extend our understanding to \emph{heterogeneous} graphs. Heterogeneity may be introduced in a number of ways. One of the most common is considering graphs with vertices not all of the same degree. Previous work has shown that this distribution of vertex degrees affects the establishment of new types. For example, a mutant type may have an advantage if it  appears on a high-degree vertex while the population is undergoing the death-birth process and a 
disadvantage on the same vertex under the birth-death process \cite{antalrednersood06, broomrychtarstadler11, maciejewski13}. Environmental evolutionary graphs allow for another type of heterogeneity, one that does not depend on the degree of a vertex: an individual experiences an increase in fecundity simply if its type matches that of the vertex on which it resides. There is no reason to suppose an advantageous mutant is advantageous everywhere in the environment. A type of individual may flourish in one part of environment and flounder in another. Environmental evolutionary graphs are a convenient abstraction of this notion of location-dependent advantage. 

There are a few obvious extensions of the current work. One is to extend the current setup to include games played on environmental evolutionary graphs. There are a multitude of ways that this could be done. For example, each individual may have a baseline fecundity that depends on their location in the environment. Added to this is the payoff garnered from their game interactions. This could lead to variation in how the game affects the fitness of an individual: it is expected that in ``poor'' sites the game will matter more than in ``good'' sites. This is analogous to varying the selection strength, a factor known to affect the outcome of a game \cite{wualtrockwangtraulsen10}. Another possibility is that vertices could be thought of containing only so much of a resource and the individual occupying that vertex must decide between sharing or hoarding. 



Environmental evolutionary graphs are also interesting from a purely mathematical perspective. As was shown here, the birth-death process on properly two-coloured graphs does not depend on $r$. Even though it is doubtful that a ``properly two-coloured'' environment exists in nature, it is of interest to check if this is the largest class of graphs on which the birth-death process is independent of $r$. We also gave an example of a graph on which all colours have the same expected long-term share of the population. Is it possible to classify all such graphs? Also, certain colourings of environmental evolutionary graphs were shown here to decrease the time taken for a mutant invader to establish in the population. A general theory of population structures that minimize the time to fixation would be very interesting and may prove to have applications to populations management and the spread of disease on social or contact networks. 


\bibliographystyle{plain}
\bibliography{thesis}

\begin{thebibliography}{10}

\bibitem{antalrednersood06}
T.~Antal, S.~Redner, and V.~Sood.
\newblock Evolutionary dynamics on degree-heterogeneous graphs.
\newblock {\em Physical Review Letters}, 96:188104--1--4, 2006.

\bibitem{antalscheuring06}
T.~Antal and I.~Scheuring.
\newblock Fixation of strategies for an evolutionary game in finite
  populations.
\newblock {\em Bulletin of Mathematical Biology}, 68:1923--1944, 2006.

\bibitem{barabasialbert99}
L.~Barab\'{a}si and R.~Albert.
\newblock Emergence of scaling in random networks.
\newblock {\em Science}, 286:509--512, 1999.

\bibitem{broomrychtarstadler11}
M.~Broom, J.~Rychtar, and B.~Stadler.
\newblock Evolutionary dynamics on graphs---the effect of graph structure and
  initial placement on mutant spread.
\newblock {\em Journal of Statistical Theory and Practice}, 5:369--381, 2011.

\bibitem{cantrellcosner91}
R.S. Cantrell and C.~Cosner.
\newblock The effects of spatial heterogeneity in population dynamics.
\newblock {\em Journal of Mathematical Biology}, 29:315--338, 1991.

\bibitem{clarkethomaselmeshochberg97}
R.~Clarke, J.~Thomas, G.~Elmes, and M.~Hochberg.
\newblock The effects of spatial patterns in habitat quality on community
  dynamics within a site.
\newblock {\em Proceeding of the Royal Society B}, 264:347--354, 1997.

\bibitem{diestel10}
R.~Diestel.
\newblock {\em Graph Theory}.
\newblock Graduate Texts in Mathematics. Springer-Verlag, fourth edition, 2010.

\bibitem{ewens04}
W.J. Ewens.
\newblock {\em Mathematical Population Genetics}.
\newblock Springer, second edition, 2004.

\bibitem{fisher30}
R.A. Fisher.
\newblock {\em The Genetical Theory of Natural Selection}.
\newblock Oxford University Press, 1930.

\bibitem{fleishman02}
E.~Fleishman, C.~Ray, P.~Sj{\"o}gren-Gulve, C.L. Boggs, and D.D. Murphy.
\newblock Assessing the roles of patch quality, area, and isolation in
  predicting metapopulation dynamics.
\newblock {\em Conservation Biology}, 16:706--716, 2002.

\bibitem{gliddonstrobeck75}
C.~Gliddon and C.~Strobeck.
\newblock Necessary and sufficient conditions for multiple-niche polymorphism
  in haploids.
\newblock {\em The American Naturalist}, 109:233--235, 1975.

\bibitem{grafen06}
A.~Grafen.
\newblock A theory of {F}isher's reproductive value.
\newblock {\em Journal of Mathematical Biology}, 53:15--60, 2006.

\bibitem{markov}
C.M. Grinstead and J.L. Snell.
\newblock {\em Introduction to Probability}, chapter~11.
\newblock American Mathematical Society, 2006.

\bibitem{hiebelermichaudwassermanbuchak13}
D.~Hiebeler, I.~Michaud, B.~Wasserman, and T.~Buchak.
\newblock Habitat association in populations on landscapes with
  continuous-valued heterogeneous habitat quality.
\newblock {\em Journal of Theoretical Biology}, 317:47--54, 2013.

\bibitem{hobbshanley90}
N.T. Hobbs and T.A. Hanley.
\newblock Habitat evaluation: Do use/availability data reflect carrying
  capacity.
\newblock {\em The Journal of Wildlife Management}, 54:515--522, 1990.

\bibitem{kimuraweiss64}
M.~Kimura and G.H. Weiss.
\newblock The stepping stone model of population structure and the decrease of
  genetic correlation with distance.
\newblock {\em Genetics}, 49:561--575, 1964.

\bibitem{latoregouldmortimer99}
J.~Latore, P.~Gould, and A.~Mortimer.
\newblock Effects of habitat heterogeneity and dispersal strategies on
  population persistence in annual plants.
\newblock {\em Ecological Modelling}, 123:127--139, 1999.

\bibitem{levene53}
H.~Levene.
\newblock Genetic equilibrium when more than one ecological niche is available.
\newblock {\em The American Naturalist}, 87:331--333, 1953.

\bibitem{levinsmacarthur66}
R.~Levins and R.~MacArthur.
\newblock The maintenance of genetic polymorphism in a spatially heterogeneous
  environment: Variations on a theme by {H}oward {L}evins.
\newblock {\em The American Naturalist}, 100:585--589, 1966.

\bibitem{levins69}
S.~Levins.
\newblock Some demographic and genetic consequences of environmental
  heterogeneity for biological control.
\newblock {\em Bulletin of the Entomological Society of America}, 15:237--240,
  1969.

\bibitem{liebermanhauertnowak05}
E.~Lieberman, C.~Hauert, and M.A. Nowak.
\newblock Evolutionary dynamics on graphs.
\newblock {\em Nature}, 433:312--316, 2005.

\bibitem{maciejewski13}
W.~Maciejewski.
\newblock Reproductive value in graph-structured populations.
\newblock {\em Under review.}, 2012.

\bibitem{maciejewskihauertfu13}
W.~Maciejewski, C.~Hauert, and F.~Fu.
\newblock Evolutionary dynamics in populations with heterogeneous structures.
\newblock {\em Under Review}, 2013.

\bibitem{moran58}
P.A.P. Moran.
\newblock Random processes in genetics.
\newblock {\em Proceedings of the Cambridge Philosophical Society}, 54:60--71,
  1958.

\bibitem{nakamarumatsudaiwasa97}
M.~Nakamaru, H.~Matsuda, and Y.~Iwasa.
\newblock The evolution of cooperation in a lattice-structured population.
\newblock {\em Journal of Theoretical Biology}, 184:65--81, 1997.

\bibitem{nowak06a}
M.~Nowak.
\newblock {\em Evolutionary Dynamics: Exploring the Equations of Life}.
\newblock The Belknap Press of Harvard University Press, 2006.

\bibitem{nowakmay92}
M.A. Nowak and R.M. May.
\newblock Evolutionary games and spatial chaos.
\newblock {\em Nature}, 359:826, 1992.

\bibitem{nowaksasakitaylorfudenberg04}
M.A. Nowak, A.~Sasaki, C.~Taylor, and D.~Fudenberg.
\newblock Emergence of cooperation and evolutionary stability in finite
  populations.
\newblock {\em Nature}, 428:646--650, 2004.

\bibitem{ohtsukihauertliebermannowak06}
H.~Ohtsuki, C.~Hauert, E.~Lieberman, and M.~A. Nowak.
\newblock A simple rule for the evolution of cooperation on graphs and social
  networks.
\newblock {\em Nature}, 441:502--505, 2006.

\bibitem{ohtsukinowak06}
H.~Ohtsuki and M.A. Nowak.
\newblock Evolutionary games on cycles.
\newblock {\em Proceedings of the Royal Society B}, 273:2249--2256, 2006.

\bibitem{roussetbilliard00}
F.~Rousset and S.~Billiard.
\newblock A theoretical basis for measures of kin selection in subdivided
  populations: Finite populations and localized dispersal.
\newblock {\em Journal of Evolutionary Biology}, 13:814--825, 2000.

\bibitem{strobeck79}
C.~Strobeck.
\newblock Haploid selection with $n$ alleles in $m$ niches.
\newblock {\em The American Naturalist}, 113:439--444, 1979.

\bibitem{szabofath07}
G.~Szab\'{o} and G.~F\'{a}th.
\newblock Evolutionary games on graphs.
\newblock {\em Physics Reports}, 446:97--216, 2007.

\bibitem{taylor90}
P.D. Taylor.
\newblock Allele-frequency change in a class-structured population.
\newblock {\em The American Naturalist}, 135:95--106, 1990.

\bibitem{taylorlillicrapcownden11}
P.D. Taylor, D.~Cownden, and T.~Lillicrap.
\newblock Inclusive fitness analysis on mathematical groups.
\newblock {\em Evolution}, 65:849--859, 2011.

\bibitem{taylordaywild07a}
P.D. Taylor, T.~Day, and G.~Wild.
\newblock Evolution of cooperation in a finite homogeneous graph.
\newblock {\em Nature}, 447:312--316, 2007.

\bibitem{traulsenhauert09}
A.~Traulsen and C.~Hauert.
\newblock Stochastic evolutionary game dynamics.
\newblock In H.G. Schuster, editor, {\em Reviews of Nonlinear Dynamics and
  Complexity}. Wiley-VCH, 2009.

\bibitem{wright31}
S.~Wright.
\newblock Evolution in {M}endelian populations.
\newblock {\em Genetics}, 16:97--159, 1931.

\bibitem{wualtrockwangtraulsen10}
B.~Wu, P.~Altrock, L.~Wang, and A.~Traulsen.
\newblock Universality of weak selection.
\newblock {\em Physical Review E}, 82:046106--1--11, 2010.

\bibitem{yamanaka09}
T.~Yamanaka, K.~Tanaka, K.~Hamasaki, Y.~Nakatani, N.~Iwasaki, D.~S. Sprague,
  and O.~N. Bj{\o}rnstad.
\newblock Evaluating the relative importance of patch quality and connectivity
  in a damselfly metapopulation from a one-season survey.
\newblock {\em Oikos}, 118:67--76, 2009.

\end{thebibliography}

\section{Appendix}

\subsection{Formal Definition of Environmental Evolutionary Graphs.}

Our intention in this first appendix is to place environmental evolutionary graph theory on a rigorous footing. We restrict our attention to two-coloured environmental evolutionary graphs for simplicity. All the following results can be extended to multi-coloured graphs. 

Let $G$ be any finite connected graph and let $\bgcolor$ be any function from $V(G)$ into $\{R, B\}$. We will regard $\bgcolor$ as the fixed \emph{background coloring} of $G$.

Let $\statespace$ be the set of all functions from $V(G)$ into $\{R,B\}$. These functions are the \emph{foreground colourings} of $G$. It is a simple exercise to verify that there are $2^{|G|}$ functions in $\statespace$.

Given any $S \in \statespace$ and any $w \in V(G)$, we may wish to talk about the state obtained by switching the colour of one vertex and leaving the rest
alone. Hence we define $\byflipping{S}{w}$ to be the state given by 
\[ [\byflipping{S}{w}](v) = \begin{cases} S(v) & \text{if $v \neq w$} \\
\red & \text{if $v = w$ and $S(w) = \blue$} \\ \blue & \text{if $v = w$ and $S(w) = \red$.}
\end{cases}
\]
For any $v \in V(G)$, we will use the notation $N(v)$ to refer to the set of neighbors of $v$, and for each $S \in \statespace$ we similarly define $\enemies{S}{v}$
to be the set of opposite-colour neighbors of $v$, given by \[ \enemies{S}{v} = \{w \in N(v) : S(w) \neq S(v) \}. \] 

We are now ready to define our transition matrix. Let $\transmatrix = [P_{ij}]$ be the $2^{|G|} \times 2^{|G|}$ matrix indexed by the states $S$ of the population with entry $i,j$ given by the probability $P_{ij}$ that the population transitions from state $i$ to state $j$. We wish to use $\transmatrix$ as the transition matrix for our Markov chain. In order to do this, we must prove the following lemma:
\begin{lemma}
$\transmatrix$ is well-defined and stochastic.
\end{lemma}
\begin{proof}
The only way $\transmatrix$ could fail to be well-defined is if $\byflipping{S}{v} = S$ for some $S,v$ or if $\byflipping{S}{v} = \byflipping{S}{w}$ for some $v \neq w$. It follows immediately from the definition of $\oplus$ that neither of these conditions can obtain, so that $\transmatrix$ is well-defined. By definition, the rows of $\transmatrix$ sum to $1$, and all its entries are nonnegative. Hence, $\transmatrix$ is stochastic.
\end{proof}

\begin{definition}
An \emph{environmental graph} is a graph $G$ equiped with a function $b: V(G)\to C$ and a real number $r \geq 1$. 
\end{definition}

We note that when we are concerned with conditional probabilities of the form $P(E | X_0 = S)$, the initial distribution $\psi$ is irrelevant, since all of the chains $G_\psi$ have the
same transition matrix $\transmatrix$. Hence we will not bother specifying an initial distribution in these circumstances: when we say that $G$ has some given property of stochastic processes, we mean that $G$ has that property for every starting vector.

We now deduce some elementary facts about the long-run behavior of the processes $G_\psi$. We suppose that $G$ initially consists of some mix of $R$ and $B$. This mix came about from the introduction of a mutant type in a pure state of the population. We suppose that the probability of mutation $\mu$ is essentially $0$ so that the population reaches a pure state before another mutation occurs. Because of this assumption, we outright ignore the mutation process for the time being.
\begin{proposition}\label{absorbing}
Let $G$ be any environmental graph. Then the pure, all-$R$ or $B$ states are absorbing in $G$. Moreover, with probability $1$, $G$ eventually reaches a pure state.
\end{proposition}
\begin{proof}
It is clear that the singleton containing any monochromatic state is a recurrent class, since if $S$ is a monochromatic state then $\enemies{S}{v}$ is empty for all $v$, so that $P_{S, T} = 0$ for all $T \neq S$.

To see that they are the only recurrent classes, let any non-monochromatic state $S \in \statespace$ be given. Suppose $S$ has $n$ blue vertices ($0 < n < |G|$). Since $G$ is connected, there exists some blue vertex $v$ with a red neighbor $w$. Since $f_{S_gw} > 0$, we see that $P_{S, \byflipping{S}{v}} > 0$, so with positive probability we may move from $S$ to a state with $n-1$ blue vertices. By the same argument, from that state we may move to one with $n-2$ blue vertices, and by induction we see that in $n$ steps we may move with positive probability to a state with $0$ blue vertices, i.e., the state $\allred$. Since $\allred$ is accessible from $S$ and $\allred$ is absorbing, we conclude that $S$ is not a recurrent state.
\end{proof}

The intuitive explanation accompanying the first half of this is obvious: since there is no mechanism for introducing genetic variation in this model, once an allele is gone it's gone for good. Hence $\allred$ and $\allblue$ are absorbing. That there are no other recurrent classes -- that extinction of one allele occurs almost always -- is less intuitively obvious, but is a standard feature of models derived from the Moran process.

We will analyze this model by considering the probability, given an initial state $S$, that we end up in the all-red state versus the probability that we end up in the all-blue state. Let $X_n \to \allred$ denote the event that $X_n = \allred$ for all sufficiently large $n$, and define $X_n \to \allblue$ similarly. We then have the following definition.

\begin{definition}
Let $G$ be an environmental graph. Then the \emph{fixation probability vector} of $G$, written $\rho(G)$ or simply $\rho$, is the vector indexed by $\statespace$ whose $S$th entry $\rho_{R|S}$ is given by \[ \rho_{R|S} = P(X_n \to \allred | X_0 = S). \]
\end{definition}

For any particular environmental graph, it is possible in principle to manually calculate $\rho$ by the known techniques for dealing with absorbing Markov chains \cite{markov}, but since the size of $\transmatrix$ grows exponentially in $|G|$, this rapidly becomes impractical. We would therefore like to determine the values of $\rho$ analytically, when this is possible.

\begin{proposition}
Let $G$ be any environmental graph. Then $\rho = \transmatrix\rho$.
\end{proposition}
\begin{proof}
Let any states $S, T \in \statespace$ be given. By elementary probability theory we have
\begin{align*}
P(\text{$X_n \to \allred$ and $X_1 = T$} | X_0 = S)
=&P(X_n \to \allred | \text{$X_1 = T$ and $X_0 = S$})\cdot\\
&P(X_1 = T | X_0 = S).
\end{align*}
By the Markov property and the definition of $\transmatrix$, this reduces to 
\[ P(\text{$X_n \to \allred$ and $X_1 = T$} | X_0 = S)
= P(X_n \to \allred | X_1 = T)P_{S,T}. \]
Since the transition probabilities are independent of $n$ and since the event $X_n \to \allred$ only depends on the infinite tail of the $X_i$, we see that \[ P(X_n \to \allred | X_1 = R) = P(X_n \to \allred | X_0 = S) = \rho_{R|S}. \] Hence
\[ P(\text{$X_n \to \allred$ and $X_1 = T$} | X_0 = S) = \rho_{R|S}P_{S,T}. \] On summing over all possible $T$ (since the events involved are clearly
mutually exclusive), we obtain \[ P(X_n \to \allred) = \sum_{T \in \statespace}\rho_{R|S}P_{S,T} = \transmatrix \rho. \]
\end{proof}
By simple algebraic manipulation of the above, we see that $\vec x$ is a solution of the linear system
\begin{eqnarray}
(I - \transmatrix)\rho = 0.
\label{eq:transprob} 
\end{eqnarray}
Since the rows of $\transmatrix$ sum to $1$, the rows of $I-\transmatrix$ sum to $0$, and so we see that the column vector whose entries are all $1$ is a solution of this system. Yet we know that $\rho_{R|B} = 0$ and $\rho_{R|R} = 1$, so that $\rho$ is linearly independent from the all-$1$ vector. The question
then naturally arises: is there a \emph{unique} (up to scaling) nonzero vector which is linearly independent of the all-$1$ vector? The following lemma answers this question in the affirmative:
\begin{lemma}\label{nullity}
Let $G$ be any environmental graph. Then the dimension of the null space of $(I-\transmatrix)$ is $2$.
\end{lemma}
\begin{proof}
On removing the rows and columns corresponding to $\allred$ and $\allblue$ from $\transmatrix$, we are left with a matrix containing only the rows and columns corresponding to the transient states of the system. In the theory of absorbing Markov chains, this matrix is known as $Q$, and it is known that $I-Q$ is invertible
\cite[pp.~418]{markov}. Since we only add two rows and columns to $I-Q$ to obtain $I-\transmatrix$, we see that $\nullity(I-\transmatrix) \leq 2$. On the other hand, since the rows corresponding to $\allred$ and $\allblue$ contain only $0$, we see that $\nullity(I-\transmatrix) \geq 2$.
\end{proof}
We therefore have the following.
\begin{lemma}\label{strategy}
Let $G$ be any environmental graph, and
let $\vec y$ be any solution to the system $(I-\transmatrix)\vec{y} = 0$ such that
$\vec y_R= 1$ and $\vec y_B = 0$. Then $\vec y$ is the fixation 
probability vector of $G$.
\end{lemma}
\begin{proof}
Since $\vec y_R \neq \vec y_B$, we see that $\vec y$ is linearly independent
of the all-$1$ vector. By Lemma~\ref{nullity}, this means that
$\vec y = c_1 \rho + c_2 \vec{1}$. Since
$\vec y_B = \rho_{R|B} = 0$, we have $c_2 = 0$; then since
$\vec y_R = \rho_{R|R} = 1$ we have $c_1 = 1$ so that $\vec y = \rho$.
\end{proof}
This defines our strategy: in order to prove that some candidate vector $\vec y$ is the absorption probability vector, we will only need to prove that it satisfies the conditions of Lemma~\ref{strategy}.

\subsection{A Mean-field Approximation}

We establish \ref{eq:fixprob} in a way similar to the proof of the fixation probability in the classical Moran process (see, \cite{moran58, nowak06a}). 
Let $i$ be the number of $R$ types on $G$. We need only the one-step transition probabilities $P_{i,i+1}$ of going from $i$ to $i+1$ and $P_{i,i-1}$ of going from $i$ to $i-1$. For the birth-death process, these are easily calculated as
\begin{eqnarray}
 P_{i,i+1} & = &  \left [ \frac{(1-d)i+rdi}{(rd+(1-d))i +(r(1-d)+d)(N-i)} \right ]\frac{N-i}{N}  \label{eq:up}\\
 P_{i,i-1} & = & \left [ \frac{r(1-d)(N-i)+d(N-i)}{(rd+(1-d))i +(r(1-d)+d)(N-i)}\right ] \frac{i}{N}. \label{eq:down}
\end{eqnarray}
Define 
\begin{eqnarray}
\label{eq:gamma}
 \gamma_i = \frac{P_{i,i-1}}{P_{i,i+1}} = \dfrac{r(1-d)+d}{(1-d)+rd}.
\end{eqnarray}
It can be shown that taking the product of the $N-1$ terms $\gamma_i$, as in \cite{nowak06a}, yields the fixation probability 
\begin{eqnarray}
\label{eq:prob}
\rho_{R|m} = \frac{\displaystyle 1+\sum_{k=1}^{m-1}\prod_{j=1}^k\gamma_m}{\displaystyle 1+\sum_{k=1}^{N-1}\prod_{j=1}^k\gamma_m}.
\end{eqnarray}
Substituting Equations (\ref{eq:up}) and (\ref{eq:down}) into Equation (\ref{eq:gamma}) and subsequently into Equation (\ref{eq:prob}) yields the approximation.

\subsection{Proof of Theorem \ref{thm:twocoloured}}
\begin{theorem*}
Given a properly two-coloured graph $G$ undergoing either a birth-death process and a set $M\subset V(G)$ of vertices occupied by $R$ (red) types then the probability $\rho_{R|M}$ that the $R$ fix in the population is
\begin{eqnarray}
\rho_{R|M} = \sum_{i \in M} \rho_{neutral|i},
\end{eqnarray}
where $\rho_{neutral|i}$ is the neutral fixation probability of a single $R$ starting at vertex $v_i$. 
\end{theorem*}
The proof of this theorem relies on the following result. 
\begin{lemma}
Let $G$ be a properly two-coloured graph and let $S$ be the state of the population. For every $v_i\in V(G)$, we have $f_i(S) = f_j(S)$ for all $j\in \mathcal{N}_S'(v_i)$. 
\end{lemma}
\begin{proof}
Since $G$ is properly two-coloured, $v_i$ and $v_j$ are of opposite colours. So, if $S(v_i)=b(v_i)$, then $S(v_j)=b(v_j)$, by virtue of $j$ being in $\mathcal{N}_S'(v_i)$. The same is true if $S(v_i)\neq b(v_i)$. 
\end{proof}
Recall that $S\oplus v_j$ is defined as the state obtained from state $S$ by switching the colour of the individual on vertex $v_j$. Now to prove Theorem \ref{thm:twocoloured}.
\begin{proof}
 Let $\rho$ be the vector of fixation probabilities indexed by $S$. Since $\rho$ is the fixation probability vector, it satisfies Equation (\ref{eq:transprob}). This yields,
\begin{eqnarray}
((I-\transmatrix)\rho)_S & = & (1-P_{S, S}) \rho_S - \sum_{T\neq S}P_{S,T}\rho_T.
\label{eq:firstlineofproof}
\end{eqnarray}
Since the population state can change by at most one vertex colour, the state $T$ is of the form $S\oplus v_j$ for some vertex $v_j$. This allows Equation (\ref{eq:firstlineofproof}) to be written
\begin{eqnarray}
& = & \rho_{S} \sum_{v_j\in V(G)} P_{S,S\oplus v_j} - \sum_{v_j \in V(G)} P_{S,S\oplus v_j}\rho_{S\oplus v_j}\nonumber \\
& = &  \sum_{v_j\in V(G)} P_{S,S\oplus v_j} (\rho_{S} -\rho_{S\oplus v_j}). 
\label{eq:midway}
\end{eqnarray}
It is at this stage of the proof that we require the population to be undergoing a birth-death process. This permits a calculation of the transition probability: 
\begin{eqnarray}
P_{S,S\oplus v_j} = \dfrac{\displaystyle \sum_{v_k \in \mathcal{N}'(v_j)}f_k}{\displaystyle \sum_{v_l \in V(G)} f_l}
\label{eq:bdtrans}
\end{eqnarray}

To proceed, notice that $v_j$ will be red in exactly one of $S$ and $S\oplus v_j$. Define
\begin{eqnarray}
\delta(v_j) = \left \{
\begin{array}{c c c}
1 & \hbox{if} & v_j = 1 \ \hbox{ in } S,\\
-1 & \hbox{if} & v_j = 1 \ \hbox{ in } S\oplus v_j.
\end{array} \right.
\end{eqnarray}
This allows for
\begin{eqnarray}
\rho_{S} -\rho_{S\oplus v_j} = \delta(v_j)  \dfrac{\displaystyle \dfrac{1}{d_j}}{\displaystyle \sum_{l \in V(G)} \dfrac{1}{d_l}}.
\end{eqnarray}
Substituting this into Equation (\ref{eq:midway}), and combining with Equation (\ref{eq:bdtrans}), yields
\begin{eqnarray}
& \displaystyle \sum_{v_j\in V(G)} P_{S,S\oplus v_j} (\rho_{S} -\rho_{S\oplus v_j}) = \\
&  \dfrac{1}{\displaystyle \sum_{v_j \in V(G)} f_j\cdot d_j}\cdot \displaystyle\sum_{v_j\in V(G)}\sum_{v_k \in \mathcal{N}'(v_j)}  \left (\delta(v_j) \displaystyle \dfrac{f_k}{d_i} \right ) \nonumber.
\label{eq:almostthere}
\end{eqnarray}
Denote the bracketed expression in Equation (\ref{eq:almostthere}) as $\tau(v_j,v_k)$. From Lemma 1,
\begin{eqnarray}
 \tau(v_j,v_k) = \delta(v_j) \displaystyle \dfrac{f_k}{d_i} = -\delta(v_k) \displaystyle \dfrac{f_j}{d_i} = -\tau(v_k,v_j),
\end{eqnarray}
for all $v_j\in V(G)$ and $v_k \in \mathcal{N}'(v_j)$. Since the sum in Equation (\ref{eq:midway}) is over all vertices, each $\tau(v_j,v_k)$ cancels with a $\tau(v_k, v_j)$. In all, 
\begin{eqnarray}
((I-P)\rho)_S, 
\end{eqnarray}
or, 
\begin{eqnarray}
 P\rho = \rho,
\end{eqnarray}
which establishes the theorem.
\end{proof}

\subsection{Calculations for Fixation Probability and Time to Fixation.}

This section focuses on the calculations needed for the fixation probability and time to fixation in the graph in Figure \ref{fig:onered}. Define the states $S_0=(0,0,0)$, $S_1 = (0,1,0)$, $S_2 = (1,0,0)$, $S_3=(1,1,0)$, and $S_4 = (1,1,1)$. These are the three possible states of the population, up to symmetry.

For the fixation probability define $\phi_i$ to be the probability that the population fixes at a state of all $R$ given that it started in state $S_i$. The $\phi_i$ satisfy the system of equations
\begin{eqnarray}
\label{eq:trans}
\phi_0 & = & 0 \nonumber \\
\phi_1 & = &  P_{1, 0}\phi_{0} + P_{1,3}\phi_3 +(1-P_{1, 0}-P_{1, 3})\phi_1,\nonumber \\
\phi_2 & = & P_{2, 0}\phi_{0} +  P_{2,3}\phi_3 +(1-P_{2, 0}-P_{2,3})\phi_2, \nonumber \\
\phi_3 & = & P_{3, 1}\phi_{1} +  P_{3,2}\phi_2 + P_{3,4} + (1-P_{3, 1}-P_{3, 2}-P_{3,4})\phi_3, \nonumber \\
\phi_4 & = & 1,
\label{eq:sys1}
\end{eqnarray}
where $P_{i,j}$ is the probability of transitioning from state $S_i$ to state $S_j$. For the population under consideration undergoing a birth-death process, 
\begin{eqnarray}
 P_{1,0}  = \dfrac{2}{2+r} ,      & P_{2,0}  =  \dfrac{1}{2\cdot3}, & P_{1,3} = \dfrac{r}{2+r}  \nonumber \\
 P_{2,3} = \dfrac{1}{2+r}, & P_{3,4} = \left(\dfrac{1}{2}\right)\dfrac{r}{2+r}, & P_{3,3} = \dfrac{1}{2+r}+\left(\frac{1}{2}\right) \dfrac{r}{2+r}.
\end{eqnarray}
These are substituted into System (\ref{eq:sys1}) and solved. The solutions are then used to generate Equation (\ref{eq:onered}) by weighting by the probability that the mutant arises on either leaf or the hub. Suppose the population initially consists of all blue individuals and a mutation occurs, producing a red offspring. This offspring appears on either leaf with probability $1/3$ and on the hub with probability $2/3$. This yields
\begin{eqnarray}
 \overline{\rho^*} & = & 2\cdot \frac{1}{3} \cdot \phi_2 + \frac{2}{3} \cdot \phi_1 \nonumber \\
\ & = & \frac{2r^2(1+r)}{2r^3+5r^2+4r+4}.
\end{eqnarray}
A similar calculation is employed to generate Equation (\ref{eq:onered}).

For the time to fixation, we use an approach similar to \cite{antalscheuring06}; see also \cite{traulsenhauert09}. Define $T_i$ to be the time the population takes to reach fixation conditioned on the event that the population reaches fixation given that it currently is in state $S_i$, where the states are as above. For the $3$-line example, the $T_i$ satisfy
\begin{eqnarray}
\label{eq:times}
\phi_1 T_1 & = &  P_{1,3}\phi_3(T_3 + 1),\nonumber \\
\phi_2 T_2 & = &  P_{2,3}\phi_3(T_3 + 1) +(1-P_{2, 0}-P_{2,3})\phi_2(T_2+1), \hbox{ and}\nonumber \\
\phi_3 T_3 & = & P_{3, 1}\phi_{1}(T_1+1) + P_{3,2}\phi_2(T_2+1)+(1-P_{3, 1}-P_{3, 2})\phi_3(T_3+1), 
\label{eq:sys2}
\end{eqnarray}
where the $\phi_i$ are as above. These solve to the equations used to generate Figure \ref{fig:fixationtimeline}.
\end{document}